\documentclass[11pt,a4paper]{article}

\usepackage{hyperref}

\usepackage{amsmath}
\usepackage[T1]{fontenc}
\usepackage[utf8]{inputenc}
\usepackage{amsfonts}
\usepackage{amsthm}
\usepackage{amstext}
\usepackage{amssymb}
\usepackage{color}
\usepackage{bbold}
\usepackage{url}
\usepackage[noend]{algpseudocode}
\usepackage[ruled,vlined,linesnumbered]{algorithm2e}
\let\oldnl\nl% Store \nl in \oldnl
\newcommand{\nonl}{\renewcommand{\nl}{\let\nl\oldnl}}% Remove line number for one line

\newcommand{\Ima}{\mathrm{Im}}
\newcommand{\rk}{\mathrm{rank}}
\newcommand{\diag}{\mathrm{diag}}

\newcommand{\cc}{\mathbb{C}}
\newcommand{\rr}{\mathbb{R}}

\newtheorem{theorem}{Theorem}

\newtheorem{lemma}[theorem]{Lemma}

\newtheorem{proposition}[theorem]{Proposition}
\newtheorem{corollary}[theorem]{Corollary}
\newtheorem{remark}[theorem]{Remark}

\newtheorem{definition}[theorem]{Definition}

\begin{document}

\title{An Efficient Uniqueness Theorem\\ for Overcomplete Tensor Decomposition} 
\author{Pascal Koiran\thanks{ENS de Lyon, CNRS, Université Claude Bernard Lyon 1, Inria, LIP, UMR 5668, 69342, Lyon cedex 07, France.}
}

\date{}

\maketitle

\begin{abstract} \small\baselineskip=9pt 
We give a new, constructive uniqueness theorem for tensor decomposition. It applies to order 3 tensors of format $n \times n \times p$
and  can prove uniqueness of decomposition for generic tensors up to rank $r=4n/3$ as soon as $p \geq 4$. One major advantage over Kruskal's uniqueness theorem is that our theorem has an algorithmic proof, and the resulting algorithm 
 is efficient.  Like the uniqueness theorem, it applies  in the range $n \leq r \leq 4n/3$.
 As a result, we obtain the first efficient algorithm for overcomplete decomposition of generic tensors of order~3.
 For instance, prior
to this work it was not known how to efficiently decompose generic tensors of format $n \times n \times n$ and rank $r=1.01n$ (or rank $r \leq (1+\epsilon) n$, for some constant $\epsilon >0$).
Efficient overcomplete decomposition of generic tensors of format $n \times n \times 3$ remains an open problem.

Our results are based on the method of commuting extensions pioneered by Strassen for the proof of his $3n/2$ lower bound on tensor rank and border rank. In particular, we rely on an algorithm for the computation of commuting extensions recently proposed in the companion paper~\cite{koi24}, and on the classical diagonalization-based  ``Jennrich algorithm''
for undercomplete tensor decomposition.

This is an updated version of a paper presented at SODA 2025. As a new result, we answer a question from that paper by giving a NP-hardness result for the computation of commuting extensions. The proof relies on a recent construction by Shitov~\cite{shitov25}. After the paper appearing in the SODA proceedings was written,  another
algorithm for the overcomplete decomposition of generic tensors of order~3 was proposed by Kothari,  Moitra and Wein~\cite{kothari24}.
\end{abstract}

\section{Introduction}

A tensor can be viewed as a multidimensional array with entries in a field~$K$. 
We will only consider tensors of order 3, i.e., elements of $K^{m \times n \times p}$.
In this paper we give a new uniqueness theorem for tensor decomposition and an efficient decomposition algorithm
derived from this theorem.
The rank of a tensor $T$ is the smallest integer $r$ such that one can write 
\begin{equation} \label{eq:decomp} 
T=\sum_{i=1}^r u_i \otimes v_i \otimes w_i
\end{equation}
for some vectors $u_i \in K^m, v_i \in K^n, w_i \in K^p$.
Such a decomposition is said to be essentially unique if it is unique up to permutation of the rank one terms $u_i \otimes v_i \otimes w_i$.
One can prove uniqueness of decomposition in fairly general situations. By contrast, a matrix of rank $r$ does not have a unique decomposition as a sum of $r$ matrices of rank 1 as soon as $r \geq 2$.
 Uniqueness of decomposition makes tensor methods particularly attractive in application areas such as statistics, machine learning or signal processing (see for instance the book~\cite{moitra18} for applications in machine learning, or the somewhat older survey~\cite{kolda09}). One downside is that computing a smallest decomposition, or indeed just the tensor rank, is an NP-hard problem~\cite{hastad90}. The best known uniqueness theorem is due to Kruskal (see~\cite{kruskal77} and Section~\ref{sec:kruskal}). It gives explicit conditions on the vectors $u_i,v_i,w_i$ 
 in~(\ref{eq:decomp}) which imply (essential) uniqueness of decomposition. For tensors of format $n \times n \times n$, these conditions can be met up to $r=\lfloor 3n/2 \rfloor-1$.
 Quantitatively, our uniqueness theorem and Kruskal's are incomparable. We only obtain uniqueness of decomposition up to
 $r= \lfloor 4n/3 \rfloor$, which is worse than Kruskal for all  $n \geq 10$. On the other hand, our theorem is stronger for tensors of
 format $n \times n \times p$ where $p$ is small. For instance, we can prove uniqueness of decomposition up to $r=4n/3$
 already for $p=4$. For tensors of format $n \times n \times 4$, Kruskal's theorem can prove uniqueness only up to
 $r=n+1$ (see Section~\ref{sec:kruskal} for details).
 
 One major advantage of our uniqueness theorem is that it has an algorithmic proof, and we will see that the resulting algorithm 
 is efficient: it runs in time polynomial in the input size ($n^2p$) for tensors of format $n \times n \times p$ (what we mean precisely by this is explained in Section~\ref{sec:complexity}). By contrast, no efficient algorithm is known for computing the unique decomposition promised by
 Kruskal's theorem.
 
 \subsection{Uniqueness theorem}

Before giving the statement of this theorem, we recall that a tensor of order~3 and format $m \times n \times p$ can be cut into $p$    ``slices.'' 
Each slice is a $m \times n$ matrix. These are the $z$-slices. One can of course cut $T$ in the two other directions into its $x$-slices and $y$-slices. In this paper we will only work with the $z$-slices. Also, we will only consider tensors of format 
$n \times n \times p$ (so each slice is a square matrix of size $n$).
\begin{theorem}[simple form of the uniqueness theorem] \label{th:unique_simple}
Consider a tensor $T=\sum_{i=1}^r u_i \otimes v_i \otimes w_i$  of format $n \times n \times p$ with $p\geq 4$. 
We assume that $T$ satisfies the following properties:
\begin{enumerate}
\item The $w_i$ are pairwise linearly independent.
\item Let $(T_1,\ldots,T_p)$ be the slices of $T$. The first slice $T_1$ is invertible.
%The span $\langle T_1,\ldots,T_p\rangle$ of the slices of $T$ contains an invertible matrix.
\item %Let  $A_i=A^{-1} T_i$ for $i=2,\ldots,p$. 
Let  $A_i=T_1^{-1}T_{i}$ for all $i=2,\ldots,p$. The two following properties hold for any 3 distinct indices $k,l,m \geq 2$:
\begin{itemize}
\item[(i)] The three linear spaces $\Ima [A_k,A_l]$, $\Ima [A_k,A_m]$, $\Ima[A_l,A_m]$ are of dimension $2(r-n)$.
\item[(ii)] The three linear spaces $\Ima [A_k,A_l]+\Ima [A_k,A_m]$, $\Ima [A_l,A_k]+\Ima [A_l,A_m]$ %and 
$\Ima [A_m,A_k]+\Ima [A_m,A_l]$ are of dimension $3(r-n)$.
\end{itemize}
\end{enumerate}
Then $\rk(T)=r$, and  the decomposition of $T$ as a sum of~$r$ rank one tensors is essentially unique. 
\end{theorem}
In this theorem we denote by $[A,B]=AB-BA$ the commutator of two matrices, and as usual we denote by $\Ima A$ the span
of the columns of~$A$. Hypothesis 2 can be satisfied only when $r \geq n$. This follows from~(\ref{eq:slices}), the expression for slices as a matrix product given in Section~\ref{sec:background}. The term ``overcomplete'' in the paper's title refers
to the fact that we focus on the case $r \geq n$. The so-called undercomplete case ($r \leq n$) is covered adequately by Kruskal's theorem or the earlier, more algorithmic work of Harshman and Jennrich (see \cite{harshman70} and Theorem~\ref{th:jennrich}). The latter work yields an efficient diagonalization-based algorithm for the decomposition of generic tensors in the undercomplete case, but this was not realized at the time (see Section~\ref{sec:unique} for details, and for a short review of the convoluted history of these results).

In Theorem~\ref{th:unique} we give a somewhat more general version of Theorem~\ref{th:unique_simple}. Roughly speaking, the role of $T_1$ can be played by any invertible matrix in the span of the slices; moreover, conditions (i) and (ii) do not have hold for {\em all} triples of indices.
It is well known that condition 1 is necessary for any tensor decomposition to be unique. This is due to the fact that the decomposition of a matrix of rank 2 as the sum of two rank 1 matrices is not unique (see for instance~\cite{landsberg09}).

{\bf Genericity.} It is clear that condition (ii) cannot be satisfied for $r>4n/3$, but it is perhaps not so clear that there are any interesting examples of tensors satisfying the conditions of Theorem~\ref{th:unique_simple}. 
In Section~\ref{sec:generic} we will give a genericity theorem showing that for any $n \leq r\leq 4n/3$,  these conditions are satisfied for ``almost all'' choices of the vectors $u_i,v_i,w_i$. This means that the $r(2n+p)$-tuple of coordinates of these vectors only needs to avoid the zero set of some non identically zero polynomial. In particular, for $K=\mathbb{R}$ or $K=\mathbb{C}$ the set of ``bad vectors'' is of measure 0; and for finite fields, most choices of the $u_i,v_i,w_i$ satisfy these conditions if $|K|$ is large enough compared to $n$ and $p$.
Throughout the paper, we use the term ``generic'' in this algebro-geometric sense. Namely, if $K$ is an infinite field,\footnote{If $K$ is finite, we can replace it by its algebraic closure.}
we say that a property generically holds true in $K^N$ if there is a nonidentically zero polynomial $P(x_1,\ldots,x_N)$ such that 
 the property holds true for all $x \in K^N$ such that $P(x_1,\ldots,x_N) \neq 0$.
 
\subsection{Decomposition algorithm.}

Our algorithm applies to tensors satisfying the same hypotheses as in the uniqueness theorem and as already pointed out, 
it is an algorithmic version of that theorem. More precisely, we give in Section~\ref{sec:gendecomp} an algorithm that runs 
under the same hypotheses as our general uniqueness theorem (Theorem~\ref{th:unique}). 
We also give in Section~\ref{sec:0decomp} a slightly simpler algorithm that runs under somewhat less general hypotheses.
It is nevertheless capable of decomposing generic tensors in the same range ($n \leq r \leq 4n/3$) as our main algorithm.
These are the first efficient algorithms for decomposition of generic tensors in the overcomplete case.
%\footnote{
 For instance, prior
to this work it was not known how to efficiently decompose generic tensors of format $n \times n \times n$ and rank $r=1.01n$ (or rank $r \leq (1+\epsilon)n$, for some constant $\epsilon >0$).
 It {\em is} known, however, that generic tensors of format $n \times n \times p$ can be decomposed up to rank $r=(n-1)^2/4$ if $p$ is large enough (namely, if $p \geq r$~\cite{johnston23}).
 Also, symmetric tensors of format $n \times n \times n$ and rank $n+k$ can be decomposed with a running time that scales like $n^k$~\cite{chen22}.
 
We therefore obtain an almost complete solution to the open problem at the end of Section~3.3 of Moitra's book~\cite{moitra18},
with the following two caveats:
\begin{itemize}
\item[(i)] We only consider tensors of format $n \times n \times p$. The fully rectangular case ($m \times n \times p$) should be treated as well.
\item[(ii)] Like our uniqueness theorem, the decomposition algorithm only applies to tensors with 4 slices or more.
The case $p=3$ is completely open. For instance, no efficient algorithm is known for the decomposition of generic tensors of 
format $n \times n \times 3$ and rank $r=1.01n$.
\end{itemize}
The case $p=2$ was investigated by Grigoriev~\cite{grigoriev78} and JaJa~\cite{jaja78}. 
It follows from their work that the rank of a generic tensor of format
$n \times n \times 2$ is equal to $n$.\footnote{This also follows from the upper bound $\lceil p/2 \rceil n$ on the generic rank of a tensor of format $n \times n \times p$~\cite{atkinson80},  from our characterization of tensor rank in terms of commuting matrices (Theorem~\ref{th:constr_intro}) and even from the more elementary Theorem~\ref{th:indordi}.} As a result, this case can be handled by an algorithm for undercomplete tensor decomposition such as the %Jennrich 
simultaneous diagonalization algorithm presented in Section~\ref{sec:unique}.

It should be mentioned that algorithms based on the method of sum of squares (or inspired by this method) have been proposed for overcomplete decomposition of {\em symmetric} tensors of order 3, e.g.,~\cite{hopkins16,ma16}, 
and also for ordinary tensors of order 3~\cite{kivva20}. These algorithms
are designed for the average-case setting, and therefore do not meet the goal of generic tensor decomposition.
One upside, however, is that they apply for much higher rank: up to $r \simeq n^{4/3}$ for ~\cite{hopkins16}, and up to $r \simeq n^{3/2}$ for~\cite{kivva20,ma16}. There are also algorithms for overcomplete decomposition of generic tensors of order four~\cite{cardoso07,hopkins19} (note that~\cite{cardoso07} contains a mistake, which was corrected in~\cite{johnston23}).
A smoothed analysis of the FOOBI algorithm from~\cite{cardoso07} can be found in~\cite{bhaskara22}.

After the version of this paper appearing in the proceedings  of SODA 2025 was written,  a second
algorithm for overcomplete decomposition of generic tensors of order~3 was proposed by Kothari,  Moitra and Wein~\cite{kothari24}. Their algorithm has the advantage of applying up to a higher rank than ours: 
for any $\epsilon>0$, it applies up to $r=(2-\epsilon)n$   for tensors of format $n \times n \times p$ instead of $r=4n/3$ 
for our algorithm. Their algorithm also applies to tensors of rectangular format ($m \times n \times p$).
One downside of their algorithm is that it requires more slices. Namely, 
our algorithm applies for any $p\geq 4$, but according to~\cite[Theorem 2.8]{kothari24} 
 their algorithm can be run only for $p \geq 18$ on tensors of format $n \times n \times p$ (for $p=18$ it applies up to $r=14n/9-18^3/4$; $r=2n$ is approached asymptotically as $p$ increases).

Like~\cite{hopkins16} our algorithm is a spectral method, but it is not inspired by semidefinite programming. 
As we now explain, we take as our starting point the classical simultaneous diagonalization algorithm presented in Section~\ref{sec:unique}, and ``boost it'' so it can handle overcomplete tensor decomposition.

%As we will now explain, our uniqueness theorem and the corresponding algorithm are based on fairly elementary tools of linear algebra such as matrix diagonalization and resolution of linear systems (these two ingredients are already present in Jennrich's algorithm).
%Let us now present this approach.

\subsection{Our approach: the method of commuting extensions.} \label{sec:approach}

Our decomposition algorithm follows the philosophy of ``learning from lower bounds''. This approach has been developed in particular for the reconstruction of arithmetic circuits, see for instance~\cite{garg20} and the references there (in the literature on arithmetic circuits, tensor decomposition is known as ``reconstruction of depth-3 set-multilinear circuits'').
One starts from a circuit lower bound (or from the {\em ideas} behind the lower bound proof), and then tries to turn it into a reconstruction algorithm for the same class of circuits. Usually, the lower bound is based on the method of partial derivatives~\cite{NW96} or a variant of this method. In this paper we use instead the method of commuting extensions, which was invented by Strassen to prove the following lower bound on tensor rank. This is the first time that this method is used to design a reconstruction algorithm.
\begin{theorem}[Strassen~\cite{strassen83}] \label{th:strassen}
Let $T$ be a  tensor of format $n \times n \times 3$  
with slices $T_1,T_2, T_3$.
If $T_1$ is invertible,
\begin{equation} \label{eq:strassen}
\rk(T) \geq n + \frac{1}{2} \rk(T_2T_1^{-1}T_3 - T_3T_1^{-1}T_2).
\end{equation}
\end{theorem}
Strassen also showed  that the right hand side is a lower bound on the {\em border rank} of $T$, but we will not work with 
border rank in this paper. The connection of our results to lower bounds should be clear just from the statement of Theorem~\ref{th:unique_simple}: the conclusion that $\rk(T) = r$ really boils down to the lower bound $\rk(T) \geq r$ since $\rk(T) \leq r$ is an assumption of the theorem. Note also that this lower bound follows from Theorem~\ref{th:strassen} and hypothesis 3.(i) of Theorem~\ref{th:unique_simple}. There has been one previous attempt to turn Theorem~\ref{th:strassen} into a reconstruction algorithm, but the resulting algorithm remains conjectural~\cite[Chapter 3]{persu18}.

The rank of the matrix on the right hand side of~(\ref{eq:strassen}) can be viewed as a quantitative measure of the lack of
commutativity of the matrices $T_1^{-1}T_2$ and $T_1^{-1}T_3$.  As shown by Strassen, the whole right hand side then gives a lower bound
on the size of a {\em commuting extension} of this pair of matrices.
\begin{definition}
A tuple $(Z_1,\ldots,Z_p)$  of matrices in $M_r(K)$ is said to be a {\em commuting extension} of a tuple $(A_1,\ldots,A_p)$ of matrices in $M_n(K)$ if the $Z_i$ pairwise commute 
and each $A_i$ sits in the upper left corner of a block decomposition of  $Z_i$, i.e.,
\begin{equation} \label{eq:block}
Z_i=
\begin{pmatrix}
A_i & B_i \\
C_i & D_i
\end{pmatrix}
\end{equation}
for some matrices $B_i \in M_{n,r-n}(K)$, $C_i \in M_{r-n,n}(K)$ and $D_i \in M_{r-n}(K)$.
\end{definition}
Here we denote by $M_{r,s}(K)$ the set of matrices with $r$ rows, $s$ columns and entries from $K$.
Also, $M_r(K)=M_{r,r}(K)$ as usual.
The term ``commuting extension'' was apparently coined by Degani, Schiff and Tannor~\cite{degani05}, 
who were unaware of Strassen's work.
Any tuple of matrices of size $n$ has a commuting extension of size $2n$. Indeed, one can take
\begin{equation} \label{eq:2n}
Z_i = \begin{pmatrix}
A_i & -A_i\\
A_i & -A_i
\end{pmatrix}
\end{equation}
since $Z_i Z_j=0$ for all $i,j$. It turns out that the commuting extensions that are relevant for the study of tensor rank are
those that are {\em diagonalizable} in the sense that each $Z_i$ is a diagonalizable matrix (this implies that the $Z_i$ are simultaneously diagonalizable). The commuting extension in~(\ref{eq:2n}) is certainly not diagonalizable (except if $A_i=0$ for all $i$) since the $Z_i$ are nilpotent.

 In this paper we build on Strassen's work to obtain the following characterization of tensor rank in terms of commuting extensions.
 \begin{theorem} \label{th:constr_intro}
Let $T$ be a tensor 
of format $n \times n \times p$ over an infinite field~$K$.
Assume moreover that the span $\langle T_1,\ldots,T_p \rangle$ of the slices  of $T$ 
 contains an invertible matrix.
 
 For any integer $r \geq n$,  $\rk(T) \leq r$ if and only if there exists an invertible matrix $A \in \langle T_1,\ldots,T_p \rangle$ such that the $p$-tuple of matrices $A^{-1} T_i$ admits 
a commuting extension $(Z_1,\ldots,Z_p)$ where the $Z_i$ are diagonalizable matrices of size $r$.
\end{theorem}
Other characterizations of tensor rank in the same style and related bibliographic references can be found in~\cite{koiran20commuting}. 
This preprint also provides a self-contained exposition of Strassen's lower bound.
Theorem~\ref{th:constr_intro} forms the basis of our uniqueness theorem and decomposition algorithm.
In particular, the decomposition algorithm works roughly as follows:
\newpage
\begin{enumerate}
\item Assume for simplicity\footnote{This turns out to be generically true, see Section~\ref{sec:0decomp}.} that one may take $A=T_1$ in Theorem~\ref{th:constr_intro}, and compute $A_i=T_1^{-1}T_{i}$
for all $i=2,\ldots,p$.
\item Compute a commuting extension $(Z_2,\ldots,Z_p)$ of size $r$ of $(A_2,\ldots,A_p)$  with the algorithm from the companion paper~\cite{koi24}.
\item Let $T'$ be the tensor of format $r \times r \times p$ with slices $(I_r,Z_2,\ldots,Z_p)$.
With %Jennrich's 
the simultaneous diagonalization algorithm, decompose $T'$ as a sum of $r$ tensors of rank one. 
\item Extract a decomposition of $T$ from this decomposition of $T'$.
\end{enumerate}
Uniqueness of decomposition then follows from the uniqueness of the commuting extension at step 2 (established in~\cite{koi24}) and the uniqueness of the decomposition at step 3 guaranteed by the uniqueness theorems of Jennrich and Kruskal. The actual proof of the uniqueness theorem is more involved than one might infer from the above one-line
proof sketch, in particular because the uniqueness at step 2 and step 3 is only an {\em essential} uniqueness, and because we prove  a more general version of the uniqueness theorem than in Theorem~\ref{th:unique_simple}.

%\subsection{Model of computation}
\subsection{Complexity of decomposition algorithms.} \label{sec:complexity}

Tensor decomposition algorithms can be analyzed at various level of abstraction. Let us take the example of %Jennrich's 
the simultaneous diagonalization algorithm as described in Section~\ref{sec:unique}. The algorithm performs a polynomial number of arithmetic operations and comparisons, as well as two matrix diagonalizations. At a high level of abstraction one may take arithmetic operations, comparisons and matrix diagonalization as atomic operations and conclude that %Jennrich's 
this algorithm runs in ``polynomial time.'' 
This approach is commonplace in the tensor decomposition literature, and in particular this is how %Jennrich's 
the simultaneous diagonalization algorithm is treated in Moitra's book~\cite{moitra18} and in~\cite{bhaskara14}.
%(where it is called the "simultaneous diagonalization  algorithm"). 
In arithmetic circuit reconstruction, it is likewise rather common to treat the computation of polynomial roots\footnote{Note that this is essentially equivalent to computing matrix eigenvalues.}  as an atomic operation (see~\cite{koiran19derand,KoiranSaha23} for an alternative treatment). In this paper we do for the most part a high level analysis.
As is already apparent from Section~\ref{sec:approach}, our algorithm consists essentially of a call to the algorithm from~\cite{koi24} to compute a commuting extension, followed by a call to %Jennrich's
the simultaneous diagonalization algorithm. The algorithm from~\cite{koi24} performs polynomially many arithmetic operations and comparisons (and no computation of eigenvalues). We can therefore conclude that our tensor decomposition algorithm runs in polynomial time in the above abstract sense.

Treating matrix diagonalization as an atomic operation is certainly questionnable, especially for fields such as $\mathbb{R}$ 
or $\mathbb{C}$ where the goal is usually to compute eigenvalues and eigenvectors up to some given precision $\epsilon >0$.
The complexity of a matrix diagonalization algorithm will then depend not only on the matrix size, 
but also on~$\epsilon$ and possibly on other parameters such as eigenvalue gaps.\footnote{See~\cite{banks20} for the state of the art in the finite arithmetic model and for previous results on the complexity of matrix diagonalization.}
For an example of analysis at a finer level of detail see~\cite{KS23b,KS24}, which deals with a symmetric version of %Jennrich's 
the simultaneous diagonalization algorithm. Its complexity is analyzed as a function of the tensor size $n$, its condition number and the desired precision
 $\epsilon$.\footnote{Some papers have also studied the related but different question of stability to input noise, in the setting of smoothed analysis~\cite[Theorem 1.5]{bhaskara14} or even for worst-case perturbations to the input tensor~\cite[Theorem 2.3]{bhaskara14}. See also~\cite{bhaskara22,ma16}.}
 Such a detailed analysis for the problem considered in this paper is left to future work. Needless to say, a detailed analysis can be done only when a high level description such as the one provided here has become available.
 Note also that the high level approach allows for a uniform treatment of all fields.
 
The situation is simpler for $K=\mathbb{Q}$, i.e., when the vectors $u_i,v_i,w_i$ in the decomposition of $T$ have rational entries.\footnote{Sometimes it is necessary to look for a decompostion of a tensor $T \in \mathbb{Q}^{n \times n \times p}$ in a field extension, but this is not the situation that we are considering now.} 
In this case, we show that our decomposition algorithms run in polynomial time in the bit model of computation (see Theorems~\ref{th:bit0decomp} 	and~\ref{th:bitdecomp} in Section~\ref{sec:decomp}). 
%Recall from Section~\ref{sec:approach} that we compute a commuting extension 
 %at step 2. As pointed out in~\cite{koi24}, that algorithm runs in polynomial time in the bit model of computation. 
 %Then at step 3 we call Jennrich's algorithm to compute a decomposition of the form 
 %$T'=\sum_{i=1}^r u'_i \otimes v'_i  \otimes w_i$ where $u'_i,v'_i \in K^r$. The fact that these vectors have their entries in $K$ follows from Proposition~\ref{prop:constr}.(ii). Moreover, when $K=\mathbb{Q}$ Jennrich's algorithm runs in polynomial time in the bit model of computation (Remark~\ref{rem:jennrich}). Step 3 of our algorithm therefore also runs in polynomial time in the bit model, and Step 4 is a simple algebraic manipulation (see Algorithm~\ref{algo:0decomp} for details). Hence it follows
% that our decomposition algorithm runs in polynomial time in the bit model of computation for $K = \mathbb{Q}$.
 In particular this implies that the bit size of the output (the list of vectors $u_i,v_i,w_i$) is polynomially bounded in the input size (namely, the bit size of $T$). This was a priori  not obvious from~(\ref{eq:decomp}). This conclusion applies to any tensor satisfying the hypotheses of our uniqueness theorems (Theorem~\ref{th:unique_simple}, or more generally Theorem~\ref{th:unique}).

\subsection{Worst case complexity.}

Computing the tensor rank is NP-hard in in the worst case~\cite{hastad90,schaefer16,shitov16}. In order to obtain efficient decomposition algorithms, we therefore need to make some assumptions on the input tensor (e.g., the genericity hypothesis from the present paper).
Given the close connection between tensor rank and commuting extensions, it is natural to ask whether computing the
size of the smallest commuting extension is NP-hard. Recall  from Section~\ref{sec:approach} that this problem comes in two flavors: one can consider arbitrary commuting extensions, or only those that are diagonalizable. In Section~\ref{sec:worst} 
we use a construction by Shitov~\cite{shitov25}  to show that the second version of this problem is indeed NP-hard (this appeared as an open problem in previous versions of this paper).
We also point out that the complexity of tensor rank for a constant number of slices seems to be unknown. 
For instance, is it NP-hard to compute the rank of an order 3  tensor with 3 or 4 slices?

\subsection{Organization of the paper.}

We present some background on tensor rank in Section~\ref{sec:background}.
In Section~\ref{sec:char} we give a proof of Theorem~\ref{th:constr_intro}.
As a first application of this characterization of tensor rank, we give in Theorem~\ref{th:strassenbis} 
a variation on Strassen's lower bound (Theorem~\ref{th:strassen}).

Section~\ref{sec:uniqueth} is devoted to uniqueness theorems. We first review the uniqueness theorem for commuting extension from the companion paper~\cite{koi24}, and then state the general version of our uniqueness theorem (Theorem~\ref{th:unique}).
Then we present Jennrich's uniqueness theorem and the corresponding decomposition algorithm.
Next, we review Kruskal's uniqueness theorem for the sake of completeness. This result is not used in any of our proofs.
Finally, our uniqueness theorem is established in Section~\ref{sec:uniqueproof}.

We present our two decomposition algorithms in Section~\ref{sec:decomp}, and in Section~\ref{sec:generic} we show that the hypotheses of our uniqueness theorems and decomposition algorithms are generically satisfied for input tensors in the
range $n \leq r \leq 4n/3$. We conclude the paper with  Section~\ref{sec:worst}
and some observations on the worst-case complexity of tensor rank
and commuting extensions.

%\section{Background on tensor rank}
\section{Preliminaries} \label{sec:background}

Recall from the introduction that an order 3 tensor is an element of $K^{m \times n \times p}$ where $K$ is some arbitrary field.
In many applications, the field of interest is $K=\rr$ or $K=\cc$ (the case of finite fields is also interesting for the study of matrix multiplication and other bilinear maps).
 In most of the paper we assume that $m=n$, but this section deals with the general rectangular case.
%We will only consider tensors of order 3, i.e., elements of $K^{m \times n \times p}$.
%We will sometimes work with square tensors, for which $m=n=p$.
We denote by $M_{n,p}(K)$ the set of matrices with $n$ rows, $p$ columns and entries in~$K$. 
We denote by $M_n(K)$ the set of square matrices of size $n$, by~$GL_n(K)$ the group of invertible matrices of size $n$,
and by $I_n$ the identity matrix of size $n$.

Given 3 vectors $u \in K^m, v \in K^n, w \in K^p$ we recall that their tensor product $u \otimes v \otimes w$ is the tensor of  format $m \times n \times p$
with entries: $T_{ijk}=u_i  v_j w_k$.
By definition, a tensor of this form with $u,v,w \neq 0$ is said to be of {\em rank one}.
The rank of an arbitrary tensor $T$ is defined as the smallest integer $r$ such that $T$ can be written as a sum of $r$ tensors of rank one (and the rank of $T=0$ is 0). 
An elementary counting of the number of independent parameters in such a decomposition shows that most of the square tensors of size $n$ must be of rank $\Omega(n^2)$.  For $K = \cc$, the value of the generic rank 
is known exactly: it is equal to 5 for $n=3$~\cite{strassen83} and to $\lceil n^3/(3n-2) \rceil$ for $n \neq 3$~\cite{lickteig85}. The same paper also gives a simple formula
for the dimension of the set of tensors of rank at most~$r$.
For tensors of order 4 or more, the exact value of the generic rank is not known in general. The existing  results as of 2012 are summarized in~\cite[section~5.5]{landsbergTensors}.

Matrix rank is much better behaved than tensor rank. In particular, for any $r$ the set of matrices $M \in M_{m,n}(K)$ 
such that $\rk(M) \leq r$ is a closed subset of $M_{m,n}(K)$. This well known property follows the characterization of matrix rank by vanishing minors. It is often referred to as ``lower semicontinuity of matrix rank'', and will be used throughout the
paper. By contrast, there is no such closure property for tensor rank.\footnote{This is the motivation for the introduction
of border rank. This notion will not be used in the present paper.}

As already explained in the introduction, we can cut a tensor $T \in K^{m \times n \times p}$ into  $p$ slices.  Each slice is a  $m \times n$ matrix. This will allow us to study the properties of $T$ with tools from linear algebra. It is therefore  important
to know what the slices of a tensor look like given a decomposition 
\begin{equation} \label{eq:decomp2} 
T=\sum_{i=1}^r u_i \otimes v_i \otimes w_i
\end{equation}
as a sum of rank-1 tensors.
First, we note that the $k$-th slice of $ u_i \otimes v_i \otimes w_i$ is the rank-one matrix $w_{ik} (u_i v_i^T)$.
As a result, the $k$-th slice of $T$ is
$$Z_k=\sum_{i=1}^r w_{ik} (u_i v_i^T).$$
This can be written in more compact notation:
\begin{equation} \label{eq:slices}
Z_k  = U^T D_k V
\end{equation}
where $U$ is the matrix with the $u_i$ as row vectors, $V$ is the matrix with the $v_i$ as row vectors and $D_k$ is the diagonal matrix 
$\diag(w_{1k},\ldots,w_{rk})$. We record this observation in the following proposition.
\begin{proposition} \label{prop:slices}
For a tensor $T$ of format $m \times n \times p$ we have $\rk(T) \leq r$ iff there are diagonal matrices $D_1,\ldots,D_p$ of size $r$ and two matrices $U \in M_{r,m}(K)$, $V \in M_{r,n}(K)$ such that the slices of $T$ satisfy $Z_k  = U^T D_k V$ for $k=1,\ldots,p$.
 In this case, we have a decomposition of $T$ as in~(\ref{eq:decomp2}) where the $u_i$ are the rows of $U$, the $v_i$ are the rows of $V$ and  
$D_k = \diag(w_{1k},\ldots,w_{rk})$, i.e., the $k$-th coordinate of~$w_i$ is the $i$-th diagonal entry of $D_k$.
\end{proposition}
We further record three consequences of Proposition~\ref{prop:slices}. 
%{\bf for application to constructive version.)}
\begin{corollary} \label{cor:mult}
Let $T$ be a tensor of format $m \times n \times p$ with slices $Z_1,\ldots,Z_p$. For any matrix $A \in M_{m',m}(K)$,
we have $\rk(T') \leq \rk(T)$ where $T'$ is the  tensor of format $m' \times n \times p$ and slices $AZ_1,\ldots,AZ_p$.
In particular, if $m=m'$ and $A$ is invertible then $\rk(T') =  \rk(T)$.
\end{corollary}
\begin{proof}
%Take $r=\rk(T)$ in
From~(\ref{eq:slices}) we have $Z'_k=(UA^T)^TD_kV$ for all $k$. \end{proof}
Note that this transformation of $T$ into $T'$ changes only the vectors $u_1,\ldots,u_r$. One could also multiply the
slices of $T$ from the right without increasing $\rk(T)$, and this would only change $v_1,\ldots,v_r$.

\begin{corollary} \label{cor:add}
From a tensor $T$ of format $m \times n \times p$, construct a tensor $T'$ of format $m \times n \times (p+1)$ by adding 
any slice lying in the space spanned by the slices of $T$. Then we have $\rk(T)=\rk(T')$.
\end{corollary}
\begin{proof}
Suppose that the additional slice is the linear combination $Z_{p+1}=\sum_{k=1}^p \alpha_k Z_k$, and let $r=\rk(T)$.
Using~(\ref{eq:slices}) again, we see that  $Z_{p+1}=U^TD_{k+1}V$ where $D_{k+1}=\sum_{k=1}^p \alpha_k D_k$.
Hence $\rk(T') \leq \rk(T)$ by Proposition~\ref{prop:slices}, and the converse inequality holds as well 
since $T$ is a subtensor of~$T'$.
\end{proof}

\begin{corollary} \label{cor:dimspan}
Suppose that $T$ has a decomposition as a sum of $r$ rank-1 tensors as in~(\ref{eq:decomp}). 
Then the span $\langle Z_1,\ldots,Z_p\rangle$ of the slices of $T$ is of dimension at most $r$.
Moreover, if this space is of dimension $r$ the vectors $w_i$ in~(\ref{eq:decomp}) must be pairwise linearly independent.
\end{corollary}
\begin{proof}
By~(\ref{eq:slices}), all matrices in the span are of the form $U^TDV$ where $D$ is a diagonal matrix of size $r$.
This implies the first part of the Corollary.

For the second part  we prove the contrapositive. Suppose that two of the $w_i$ are colinear, for instance $w_1=\lambda w_2$. We have a corresponding relation of proportionality for the
diagonal matrices in~(\ref{eq:slices}): the first diagonal entry of  each $D_j$ is equal to $\lambda$ times its second diagonal entry. Thus $D_1,\ldots,D_p$ would span a subspace of dimension at most $r-1$, and the span of the slices would also be of dimension $r-1$ at most.
\end{proof}

A  tensor can be naturally interpreted as the array of coefficients of the trilinear form in $m+n+p$ variables: 
\begin{equation} \label{eq:multilinear}
t(x_1,\ldots,x_m,y_1,\ldots,y_n,z_1,\ldots,z_p)=\sum_{i,j,k} T_{ijk}x_iy_jz_k.
\end{equation}
For a decomposition of $T$ as in~(\ref{eq:decomp}) we have for the corresponding  trilinear form the decomposition:
\begin{equation} \label{eq:setmulti}
t(x,y,z)=\sum_{i=1}^r (u_i^Tx)(v_i^Ty) (w_i^Tz).
\end{equation}
Such an expression is sometimes called a ``set-multilinear depth-3 homogeneous arithmetic circuit''~\cite{NW96,raz13}.

\section{A characterization of tensor rank}
\label{sec:char}

In this section we give a proof of Theorem~\ref{th:constr_intro}, 
and as a first application we give in Theorem~\ref{th:strassenbis} 
a variation on Strassen's lower bound (Theorem~\ref{th:strassen}).
%we generalize the notion of {\em independent decomposition} from Section~\ref{sec:ind} to the setting of ordinary tensors.
%As it turns out, we only need to assume that each of the two families $(u_1,\ldots,u_r)$,
%$(v_1,\ldots,v_r)$ in~(\ref{eq:decomp})  is made of linearly independent vectors; no assumption is made about the 
%family $(w_1,\ldots,w_r)$. 
% Then we give a characterization of the set of decomposable tensors, and we use it to prove Theorems \ref{th:embedding} and \ref{th:indintro}.

%\subsection{Independent decompositions of ordinary tensors}

\subsection{The case $r=n$.}

We first consider the case of a tensor of format $r \times r \times p$ and rank $r$. 
\begin{lemma}[simultaneous diagonalization by equivalence] \label{lem:simdiag}
  Let $A_1,\ldots,A_k$ be 
  matrices of size $n$
  and assume that their span contains an invertible matrix~$A$.
  There are   diagonal matrices $D_i$ and two nonsingular matrices
  $P,Q \in M_n(K)$ such that 
$A_i = P D_i Q$ for all $i=1,\ldots,k$ if and only if the $k$ matrices $A^{-1}A_i$ ($i=1,\ldots,k$) form a commuting family of diagonalizable matrices.
\end{lemma}
\begin{proof}
We will show this for the special case $A=A_1$. The general case follows easily since the tuple $(A_1,\ldots,A_k)$ is simultaneously diagonalizable by equivalence if and only if the same is true of the tuple $(A,A_1,\ldots,A_k)$.
  
  Suppose first that $A_i = P D_i Q$ where $P,Q$ are nonsingular and the $D_i$ diagonal. Since $A_1$ is invertible, the same is true of
  $D_1$ and we have $A_1^{-1}A_i = Q^{-1}D_1^{-1}D_iQ$.
  These matrices are therefore diagonalizable, and they pairwise commute.

  Assume conversely that the $A_1^{-1}A_i$ form a commuting family of diagonalizable matrices. This is well known to be a necessary and sufficient condition
  for simultaneous diagonalization by similarity~\cite[Theorem~1.3.21]{horn13}: there must exist a nonsingular
  matrix $Q$ and diagonal matrices $D_2,\ldots,D_k$ such that
  $A_1^{-1}A_i = Q^{-1}D_iQ$ for $i=2,\ldots,k$. Let $P=A_1Q^{-1}$.
  We have $A_1=PD_1Q$ where $D_1$ is the identity matrix, and for $i \geq 2$
  we have $PD_iQ = A_1Q^{-1}D_iQ = A_1(A_1^{-1}A_i)= A_i.$
\end{proof}

\begin{theorem} \label{th:indordi}
Let $S$ be a tensor 
of format $r \times r \times p$ over $K$, with an invertible matrix~$Z$ in the span of its slices.
The following properties are equivalent:
\begin{itemize}
\item[(i)] There is a family  $(w_1,\ldots,w_r)$ of vectors of $K^p$ and  two linearly independent families $(u_1,\ldots,u_r)$, 
$(v_1,\ldots,v_r)$   of vectors of $K^r$ such that
$$S=\sum_{i=1}^r u_i \otimes v_i  \otimes w_i.$$
\item[(ii)]  The $p$ matrices  ${Z}^{-1}Z_i$ commute and are diagonalizable over $K$.
\end{itemize}
\end{theorem}
\begin{proof}
By Proposition~\ref{prop:slices}, (i) is equivalent to the existence of two invertible matrices $U$ and $V$ and of diagonal matrices $D_1,\ldots,D_p$ of size $r$
such that the $z$-zlices of $S$ satisfy $Z_i=U^TD_iV$ for $i=1,\ldots,p$. This is in turn equivalent to (ii) by Lemma~\ref{lem:simdiag}.
\end{proof}
Theorem~\ref{th:indordi} shows that there is a close connection between tensor rank and simultaneous diagonalization.
This is certainly not a new observation. In particular, %Jennrich's 
the classical algorithm described in Section~\ref{sec:unique} is
essentially an algorithmic version of Theorem~\ref{th:indordi}.
Note that there is no (pairwise) linear independence assumption on the $w_i$ in Theorem~\ref{th:indordi}.
This is in contrast with the uniqueness theorems and the reconstruction algorithms.

\subsection{General case.} \label{sec:constr}

We will now derive the characterization of tensor rank in Theorem~\ref{th:constr_intro}. 
Let $T$  be a tensor of rank $r$ tensor with a decomposition
%Let $T=\sum_{i=1}^r u_i \otimes v_i \otimes w_i$.
\begin{equation} \label{eq:constrdecomp}
T=\sum_{i=1}^r u_i \otimes v_i \otimes w_i.
\end{equation}
 %In order to complete the proof of Theorem~\ref{th:constr}, we 
 Recall again from  Proposition~\ref{prop:slices} that the slices of $T$ are given by the formula $T_k = U^T D_k V$, where $U,V$ are the matrices with the $u_i$ (respectively, $v_i$) as row vectors, and $D_k$ is the diagonal matrix $\diag(w_{1k},\ldots,w_{rk})$.
As a result, the matrices in  the span $\langle T_1,\ldots,T_p \rangle$ of the slices of $T$ are exactly the matrices of the form:
\begin{equation} \label{eq:span}
A = U^TDV,\ D=\sum_{k=1}^p \lambda_k D_k
\end{equation}
where $\lambda_1,\ldots,\lambda_p \in K$. %In the proof of Proposition~\ref{prop:constr} 
We will be interested in the maximal rank of $D$ over all possible choices 
 of $\lambda_1,\ldots,\lambda_p$ in~(\ref{eq:span}). Each entry of $D$ is a linear form in the $\lambda_k$. Therefore, the maximal rank $\rho$ is equal to the number of linear forms that are not identically 0. The $i$-th linear form is identically 0 if
 $(D_k)_{ii}=w_{ik}=0$ for all $k$, that is, if $w_i=0$. The corresponding term $u_i \otimes v_i \otimes w_i$ can then
  be removed
 from~(\ref{eq:constrdecomp}) and we see that % in fact,  
 $\rk(T) \leq \rho$. Since $\rho \leq r$ by definition of $\rho$, we conclude that $\rho=\rk(T)=r$.
%The following result will be needed for the proof of Theorem~\ref{th:constr}:
We need one additional ingredient:
\begin{proposition} \label{prop:constr1}
Let $T$ be a tensor of rank $r$ and format $n \times n \times p$ over $K$ where $K$ is an infinite field.
Assume moreover %that $\rk(T) = r$ where $r \geq n$, and
 that the span $\langle T_1,\ldots,T_p \rangle$ of the 
 slices  of $T$ 
 contains an invertible matrix. 
 %Let $\rho$ be the maximal possible rank of $D$ over all possible choices 
 %of $\lambda_1,\ldots,\lambda_p$ in~(\ref{eq:span}).
  Fix a rank-$r$ decomposition: %of $T$ as in~(\ref{eq:constrdecomp}) 
 \begin{equation} \label{eq:Tdecomp}
 T=\sum_{i=1}^r u_i \otimes v_i \otimes w_i
 \end{equation} and
 pick an invertible matrix $A \in \langle T_1,\ldots,T_p \rangle$ such that the corresponding matrix $D$ in~(\ref{eq:span}) has rank $r$. The $p$-tuple of matrices $A^{-1} T_i$  must have
a commuting extension $(Z_1,\ldots,Z_p)$ where the $Z_i$ are diagonalizable matrices of size $r$. 
\end{proposition}
A proof of this proposition can be found at the end of Section~\ref{sec:constr}. In fact, we give in Proposition~\ref{prop:constr} a stronger version of this result which will be useful for our uniqueness theorem and the reconstruction algorithms.

 For the reader's convenience, let us give again the statement of  Theorem~\ref{th:constr_intro} before giving the proof:
 % give again the theorem's statement:
\begin{theorem} \label{th:constr}
Let $T$ be a tensor 
of format $n \times n \times p$ over an infinite field~$K$.
Assume moreover that the span $\langle T_1,\ldots,T_p \rangle$ of the 
 slices  of $T$  contains an invertible matrix.
 
 For any integer $r \geq n$,  $\rk(T) \leq r$ if and only if there exists an invertible matrix $A \in \langle T_1,\ldots,T_p \rangle$ such that the $p$-tuple of matrices $A^{-1} T_i$ admits 
a commuting extension $(Z_1,\ldots,Z_p)$ where the $Z_i$ are diagonalizable matrices of size $r$.
\end{theorem}

\begin{proof}
Assume first that the above commuting extension exists. Consider the tensor $S$ of format $r \times r \times (p+1)$ with slices $(Z_1,\ldots,Z_p,I_r)$.
 We can apply Theorem~\ref{th:indordi} with $Z=I_r$ and conclude that $S$ is of rank at most~$r$.
The same is true of the tensor with slices $(A^{-1} T_1,\ldots,A^{-1} T_p)$ since it is a subtensor of $S$.
 %we have just dropped the last slice of $S$.
 Hence $\rk(T) \leq r$ since multiplying each slice by an invertible matrix does not change the tensor rank (Corollary~\ref{cor:mult}).
  In the case $r=\rk(T)$, the converse follows from Proposition~\ref{prop:constr1}.\footnote{In fact, we do not just prove the existence of one suitable 
$A \in \langle T_1,\ldots,T_p \rangle$: it will be shown in the proof of Proposition~\ref{prop:constr} that most of the matrices in the span are suitable.}
The general case then follows easily: if $r> \rk(T)$ we can first construct an extension by matrices of size $\rk(T)$ using Proposition~\ref{prop:constr} and then extend it further to size $r$, e.g., by adding blocks of 0's. The resulting matrices remain commuting and diagonalizable.
\end{proof}
%Before completing the proof of Theorem~\ref{th:constr}, we record a consequence of what we have already shown.
\begin{corollary} \label{cor:constr}
Let $T$ be a tensor 
of format $n \times n \times p$ over $K$. We assume that  $A \in \langle T_1,\ldots,T_p \rangle$ is an invertible matrix such 
that the $p$-tuple of matrices $A^{-1} T_i$ admits 
a commuting extension $(Z_1,\ldots,Z_p)$ where the $Z_i$ are diagonalizable matrices of size $r$.
Fix any decomposition 
\begin{equation} \label{eq:extdecomp0}
T'=\sum_{i=1}^r u'_i \otimes v'_i \otimes w_i
\end{equation}
for the tensor $T'$ with slices  $(Z_1,\ldots,Z_p)$.
Let $V \in M_{r,n}(K)$ be the matrix having as $i$-th row the first $n$ coordinates of $v'_i$.
Let $U' \in M_{r,n}(K)$ be the matrix having as $i$-th row the first $n$ coordinates of $u'_i$, and let $U=U'A^T$. 
We have for $T$ the decomposition $T=\sum_{i=1}^r u_i \otimes v_i \otimes w_i$ where $u_1,\ldots,u_r$ are the rows
of $U$, and $v_1,\ldots,v_r$ the rows of $V$.
\end{corollary}
\begin{proof}
As we have have already seen in the above proof of Theorem~\ref{th:constr},
the existence of a   decomposition of $T'$ as in~(\ref{eq:extdecomp0}) follows from
Theorem~\ref{th:indordi}.
 The subtensor of $T'$ with slices $A^{-1}T_1,\ldots,A^{-1}T_p$ can be decomposed as  
 $\sum_{i=1}^r u''_i \otimes v_i \otimes w_i$ where the vectors $u''_1,\ldots,u''_r$ are the rows of $U'$.
 To obtain $T$ from $T'$ we multiply each slice by $A$ from the left. %By~(\ref{eq:slices}), 
 As pointed out in the proof of Corollary~\ref{cor:mult}, the effect of this on the decomposition is to multiply $U'$ by $A^T$ from the right.
\end{proof}
%Let us now resume the proof of Theorem~\ref{th:constr}.
It remains to prove Proposition~\ref{prop:constr1}. As mentioned earlier, we  will give in Proposition~\ref{prop:constr} a stronger version of this result. The first part of the following lemma was already used in~\cite{strassen83}.
\begin{lemma} \label{lem:uniqueidentity}
Let $U$ and $V$ be two $r \times n$ matrices %with $r \geq n$ 
such that $U^TV=I_n$. Then $r \geq n$, and one can add  $r-n$ columns to $U$ and $V$ in order to obtain 
two $r \times r$ matrices which satisfy $U'^TV'=I_r$. More precisely, if $U'=(U | A)$ and $V'=(V | B)$, 
 $B$ can be any matrix such that $\Ima(B) = \ker (U^T)$; and for any such $B$, $A$ is the unique matrix 
 such that $A^TB=I_{r-n}$ and $A^TV=0$.
 %In particular, for $r=n+1$ one can add one column $u$ to $U$ and one column $v$ to $V$ in order to to obtain two
%matrices of size $n+1$ which satisfy $U'^TV'=I_{n+1}$. Moreover, $u$ and $v$ are unique up to the scaling of one of  these
%two vectors by any nonzero scalar.
\end{lemma}
\begin{proof}
We have $r \geq n$ since $\rk(U^TV) \leq \max(\rk(U),\rk(V))$.
Let us add to $V$ the columns of a $r \times (r-n)$ matrix $B$ such that $\Ima(B) = \ker (U^T)$.
The resulting matrix $V'$ is invertible since $\Ima(V) \cap \ker (U^T) = \{0\}$.
Its (unique) inverse is obtained by adding to $U^T$ the rows of a matrix $A^T$ such that $\ker(A^T) = \Ima(V)$ and $A^TB=I_{r-n}$.
%Note that in order to obtain two matrices of size $r$ satisfying $U'^TV'=I_r$ we must proceed as in the above construction,
%i.e., we must add to $V$ the columns of a $r \times (r-n)$ matrix $B$ such that $\Ima(B) = \ker (U^T)$.
%In particular, for $r=n+1$ we must add a nonzero column $v \in \ker (U^T)$, and this determines $v$ up to scaling.
%Then we must add to $U$ a column $u \in \ker(V^T)$ such that $u^Tv=1$. This determines $u$ completely once $v$ is fixed.
\end{proof}

\begin{proposition} \label{prop:constr}
Let $T$ be a tensor of rank $r$ and format $n \times n \times p$ over $K$ where $K$ is an infinite field.
Assume moreover %that $\rk(T) = r$ where $r \geq n$, and
 that the span $\langle T_1,\ldots,T_p \rangle$ of the 
 slices  of $T$ 
 contains an invertible matrix. 
 %Let $\rho$ be the maximal possible rank of $D$ over all possible choices 
 %of $\lambda_1,\ldots,\lambda_p$ in~(\ref{eq:span}).
  Fix a rank-$r$ decomposition: %of $T$ as in~(\ref{eq:constrdecomp}) 
 \begin{equation} \label{eq:Tdecomp2}
 T=\sum_{i=1}^r u_i \otimes v_i \otimes w_i
 \end{equation} and
 pick an invertible matrix $A \in \langle T_1,\ldots,T_p \rangle$ such that the corresponding matrix $D$ in~(\ref{eq:span}) has rank $r$. We then have the four following properties:
 \begin{itemize}
 \item[(i)] The $p$-tuple of matrices $A^{-1} T_i$ admits 
a commuting extension $(Z_1,\ldots,Z_p)$ where the $Z_i$ are diagonalizable matrices of size $r$. 
\item[(ii)] %If $A=T_1$,  
The matrices $Z_1,\ldots,Z_p$ in this commuting extension are the slices of a tensor  $$T'=\sum_{i=1}^r u'_i \otimes v'_i \otimes w_i$$ %of format $r \times r \times p$; 
where  $v'_i \in K^r$ has its first $n$ coordinates equal to those of $v_i$, and $u'_i \in K^r$ has its first $n$ coordinates equal to 
the the $i$-th row of $UA^{-T}$ (as usual, we denote by $U$ the matrix having $u_i$ as its $i$-th row).
 %where $u'_i \in K^r$ has its first $n$ coordinates equal to those of $u_i$, and  $v'_i \in K^r$ has its first $n$ coordinates equal to those of $v_i$. 
 %moreover, the matrix $U'' \in M_{r,n}(K)$ 
 \item[(iii)] The vectors $u'_1,\ldots,u'_r$ (respectively, $v'_1,\ldots,v'_r$) are linearly independent and the span of $Z_1,\ldots,Z_p$ contains an invertible matrix.
 
 \item[(iv)] If $A=T_1$, the first slice $Z_1$ of $T'$ is equal to $I_r$.
\end{itemize}
\end{proposition}
We will see in the proof that a generic matrix  $A \in \langle T_1,\ldots,T_p \rangle$ indeed satisfies the requisite property 
$\rk(D)=r$.
Properties (ii), (iii) and (iv)  are not needed for the proof of Theorem~\ref{th:constr}, but they will be useful for  our uniqueness theorem and the reconstruction algorithms.
\begin{proof}%[Proof of Proposition~\ref{prop:constr}]
By~(\ref{eq:span}) the rank of any linear combination of slices is upper bounded by the rank of $T$, 
so we must have $r \geq n$.

We first show that %when
since  $K$ is infinite, it is actually possible to pick an invertible  matrix $A \in \langle T_1,\ldots,T_p \rangle$ such that 
$\rk(D)=r$. In order to ensure the invertibility of $A=\sum_k \lambda_k A_k$, the coefficients $\lambda_1,\ldots,\lambda_p$
must avoid the zero set of the polynomial $\det (\sum_k \lambda_k A_k)$. This polynomial is not identically 0 by  
the nonsingularity assumption on $\langle T_1,\ldots,T_p \rangle$. %The same hypothesis implies that $\rho \geq n$.
Moreover, the entries of the diagonal matrix $D$ are linear forms in $\lambda_1,\ldots,\lambda_p$ and  
each form is non-identically 0. Otherwise, as pointed out before the Proposition we would have $\rk(T) < r$ but our assumption
is $\rk(T)=r$.
The coefficients~$\lambda_k$ must therefore avoid
the union of an hypersurface of degree~$n$ and of $r$ hyperplanes.

%Note that $\rho \leq r$ since the $D_k$ are square matrices of size $r$.
% the maximal possible rank for $D$ is at most $r$. We first consider 
%As a next step, we treat the case where  $\rho=r$.
%We will also assume that $A=T_1$. 
Next we show that (i) holds true. For this we will assume that $A=T_1$.
This is without loss of generality since we can replace $T$ by the tensor of format $n \times n \times (p+1)$ with slices $A,A_1,\ldots,A_p$ (it has same tensor rank as $T$ by Corollary~\ref{cor:add}). In this way we will obtain for the $p+1$ matrices $(I_n,A^{-1}T_1,\ldots,A^{-1}T_p)$ a commuting extension  $(Z_0,Z_1,\ldots,Z_p)$.  The first matrix $Z_0$ can then simply be dropped from this list (we will see later in the proof that $Z_0=I_r$).
With an eye towards the proof of (ii) we note that this addition of  a slice to $T$ does not affect the $u_i$ or the $v_i$, and only requires the addition of one coordinate to each $w_i$ (this follows from the proof of Corollary~\ref{cor:add}).

We are therefore now in the situation where the two matrices $T_1=U^TD_1V$  and $D_1$ are invertible.
 We will even assume that $T_1=U^TD_1V=I_n$. This is again without loss of generality since we can multiply each slice of $T$ by $T_1^{-1}$ without changing $\rk(T)$. In fact, only the vectors $u_i$ need to be modified; in particular, the $w_i$  and therefore $D_1$ are unchanged (see Corollary~\ref{cor:mult} and the remark following it). With an eye towards the proof of (ii),
 we note that this transformation of the $u_i$ replaces $U$ by $UA^{-T}$. %{\bf TBC....}
 
 We therefore need to construct a commuting extension $(Z_1,Z_2,\ldots,Z_p)$ for the tuple 
 $(T_1=I_n,T_2,\ldots,T_p)$. As a further simplification, we reduce to the case where $D_1=\diag(w_{11},\ldots,w_{r1})=I_r$.
 This is possible because all the diagonal entries of $D_1$ are nonzero: we can therefore multiply each vector $w_i$ %in~(\ref{eq:constrdecomp}) 
 by $w_{i1}^{-1}$, and each %$u_i$ (or $v_i$) 
 $v_i$ by $w_{i1}$.
 This brings us to the situation where $U^TV=I_n$ as in Lemma~\ref{lem:uniqueidentity}
 since $I_n=T_1=U^TD_1V=U^TV$. 
 
 We can now construct the matrices $Z_1,\ldots,Z_p$ in the commuting extension of $(T_1=I_n,T_2,\ldots,T_p)$:
  we take them as slices of 
 a tensor $T'$ of format $r \times r \times p$ containing $T$ as a subtensor. More precisely, we will take $T'$ of the form
 \begin{equation} \label{eq:extdecomp}
 T'=\sum_{i=1}^r u'_i \otimes v'_i \otimes w_i
 \end{equation}
 where $u'_i \in K^r$ has its first $n$ coordinates equal to those of $u_i$, and  $v'_i \in K^r$ has its first $n$ coordinates equal to those of $v_i$. This indeed ensures that the slices $Z_i$ of $T'$ are extensions of the slices of $T$.
 We still have to make sure that the $Z_i$ commute and are diagonalizable.
 %Appealing again to Proposition~\ref{prop:slices}, 
 By Proposition~\ref{prop:slices} we can write $Z_i=U'^TD_iV'$
 %$Z_i=U'^TD_1V'=U'TV'$ 
 where $U',V'$ are obtained from
 $U,V$ by the addition of $r-n$ columns (containing respectively the last $r-n$ coordinates of the $u'_i$ and $v'_i$).
 By Lemma~\ref{lem:uniqueidentity}, we can choose these additional columns so that $Z_1=U'^TV'=I_r$.
 This identity implies that the rows of $U'$ are linearly independent, and likewise for $V'$. We can therefore apply Theorem~\ref{th:indordi} with $Z=Z_1$ and conclude that the  $Z_i$ indeed form a commuting family of diagonalizable matrices. This completes the proof of (i) and (iv).
 %This completes the proof of the proposition in the case $\rho=r$.
 
 It remains to check that (ii) and (iii) are satisfied as well. 
 %Let us therefore continue to assume that $A=T_1$. By~(\ref{eq:extdecomp}), $T'$ is of the form required for (ii). 
 Remember that before constructing the tensor $T'$ in~(\ref{eq:extdecomp}) we had to replace $U$ by $UA^{-T}$,
 to divide  each $w_i$ by some scalar $\alpha_{i}$ and to multiply each $v_i$ by the same scalar. 
 We can now undo the last two transformations, i.e., we divide each $v'_i$ in~(\ref{eq:extdecomp}) by $\alpha_{i}$  
 and we multiply each $w_i$ by $\alpha_i$. This leaves $T'$ unchanged, and ensures that:
 \begin{enumerate}
 \item The $w_i$ in the final decomposition of $T'$ are the same vectors $w_i$ appearing in~(\ref{eq:Tdecomp2}).
 \item  Each $v'_i$ has its first $n$ coordinates equal to those of $v_i$.
 \end{enumerate}
 This shows that (ii) holds true. As to the third property, remember that the matrices in the commuting extension are
 of the form $Z_i=U'^TD_iV'$. As a result,  $\langle Z_1,\ldots,Z_p\rangle = U'^T \langle D_1,\ldots,D_p\rangle V'$.
 We have seen that  $\langle D_1,\ldots,D_p\rangle$ contains an invertible matrix. The conclusion follows since $U'$ and $V'$ 
 are invertible matrices of size $r$, with the $u'_i$ and $v'_i$ as their respective rows.
 This completes the proof of Proposition~\ref{prop:constr}.
 \end{proof}

\subsection{A variation on Strassen's lower bound.}

As a first application of the results in Section~\ref{sec:constr} we derive a variation on Theorem~\ref{th:strassen}.
\begin{theorem} \label{th:strassenbis}
Let $T$ be a %complex 
tensor of format $n \times n \times 4$  and rank $r$,
with slices denoted  $A_1,A_2, A_3,A_4$.
Assume that  $A_1$ is invertible, and let $A'_k=A_1^{-1}A_k$ for $k=2,3,4$.
The dimension of the linear space $\Ima [A'_2,A'_3]+\Ima[A'_2,A'_4]$ is upper bounded by $3(r-n)$.
\end{theorem}
Note that each of the two spaces  $\Ima [A'_2,A'_3]$, $\Ima[A'_2,A'_4]$ has its dimension upper bounded by $2(r-n)$.
This follows from Theorem~\ref{th:strassen} applied to the two tensors with respective slices $(A_1,A_2,A_3)$ 
and $(A_1,A_2,A_4)$.  Theorem~\ref{th:strassen}  therefore implies that
$\dim(\Ima [A'_2,A'_3]+\Ima[A'_2,A'_4]) \leq 4(r-n)$. We thus improve this upper bound to $3(r-n)$.
%add
The proof relies in particular on the following lemma.
\begin{lemma}[Proposition 4 in~\cite{koi24}] \label{prop4}
If a triple $(M_1,M_2,M_3)$ of matrices of size $n$ has a commuting extension of size $r$ then we have:
$$\dim (\Ima [M_1,M_2]+ \Ima [M_1,M_3]) \leq 3(r-n).$$
\end{lemma}

If Theorem~\ref{th:strassenbis} holds for some field $K$ then it must also hold true for any subfield. For the rest of this section we will therefore assume without loss of generality that $K$ is an infinite field.
As a first step toward the proof, we show that the conclusion of Theorem~\ref{th:strassenbis} holds true if we replace 
the first slice $A_1$ by a random linear combination of slices.
\begin{lemma} \label{lem:strassenbis}
Let $T$ be as in Theorem~\ref{th:strassenbis}, and for $\lambda \in K^4$ let us denote by~$T_{\lambda}$ the linear combination of slices $\lambda_1 A_1+ \cdots + \lambda_4 A_4$. 
%Let $T_k = T_{\lambda}^{-1} A_k$ for $1 \leq k \leq 4$.
For a generic choice of  $\lambda \in K^4$, $T_{\lambda}$ is invertible and
%$$\dim(\Ima [T_2,T_3]+\Ima[T_2,4_4]) \leq 3(r-n).$$
\begin{equation} \label{eq:strassen2}
\dim(\Ima [T_{\lambda}^{-1}A_2,T_{\lambda}^{-1}A_3]+\Ima[T_{\lambda}^{-1}A_2,T_{\lambda}^{-1}A_4]) \leq 3(r-n).
\end{equation}
%where $T_k = T_{\lambda}^{-1} A_k$ for $k=2,3,4$.
\end{lemma}
\begin{proof}
It follows from Proposition~\ref{prop:constr} and its proof that for a generic  $\lambda \in K^4$, the 4-tuple 
$(T_{\lambda}^{-1}A_1,\ldots,T_{\lambda}^{-1}A_4)$ has a commuting extension $(Z_1,\ldots,Z_4)$ of size $r$. In particular, we have pointed out
in that proof that $T_{\lambda}$ is invertible for any $\lambda$ avoiding a hypersurface of degree $n$; 
and that the commuting extension can be constructed if $\lambda$ also avoids
the union of $r$ hyperplanes. Let ${\cal V}$ be the union of the hypersurface and of the $r$ hyperplanes.
For  any $\lambda \in K^4 \setminus {\cal V}$, the existence of the commuting extension implies~(\ref{eq:strassen2}) by 
a direct application of Lemma~\ref{prop4}.%~\cite[Proposition4]{koi24}.
\end{proof}
We can now give the proof of Theorem~\ref{th:strassenbis}.
\begin{proof}%[Proof of Theorem~\ref{th:strassenbis}]
The left-hand side of~(\ref{eq:strassen2}) is equal to $\rk M(\lambda)$, where $M(\lambda) \in M_{n,2n}$ has its first $n$ columns equal to $[T_{\lambda}^{-1}A_2,T_{\lambda}^{-1}A_3]$ and its last $n$ columns equal to 
$[T_{\lambda}^{-1}A_2,T_{\lambda}^{-1}A_4]$. 
In Lemma~\ref{lem:strassenbis} we have shown that $\rk M(\lambda) \leq 3(r-n)$ for $\lambda \in K^4  \setminus {\cal V}$.
We would like to show that this implies $\rk M(\lambda^1) \leq 3(r-n)$ where $\lambda_1=(1,0,0,0)$. 
Since  $K^4  \setminus {\cal V}$ is dense in $K^4$, this follows from the lower semicontinuity of matrix rank.
\end{proof}

%\newpage
\section{Uniqueness theorems} \label{sec:uniqueth}

Our results on the uniqueness of tensor decomposition rely on a uniqueness result 
from~\cite{koi24}. We first present this result.

\subsection{Uniqueness of commuting extensions.} \label{sec:uniquext}

We have defined the notion of  {\em commuting extension} in Section~\ref{sec:approach}.
Let us now recall a corresponding uniqueness result from~\cite{koi24}.
Strictly speaking, commuting extensions are never unique~\cite{degani05}. 
For $M \in GL_{r-n}(K)$, consider indeed the map $\rho_M:M_r(K) \rightarrow M_r(K)$ which sends
$Z =
\begin{pmatrix}
A & B\\
C & D
\end{pmatrix}$ to
 \begin{equation} \label{eq:action}
 \rho_M(Z) =
 \begin{pmatrix}
 I_n & 0\\
 0 & M
 \end{pmatrix}^{-1} Z \begin{pmatrix}
 I_n & 0\\
 0 & M
 \end{pmatrix} =
\begin{pmatrix}
A & BM\\
M^{-1}C & M^{-1}DM
\end{pmatrix}\end{equation}
where $I_n$ denotes the identity matrix of size $n$.
If $(Z_1,\ldots,Z_p)$ is a commuting extension of $(A_1,\ldots,A_p)$ then so is $(\rho_M(Z_1),\ldots,\rho_M(Z_p))$.
This follows immediately from the identity  
%due to the identity 
$\rho_M(ZZ')=\rho_M(Z)\rho_M(Z')$.
We say that a commuting extension of size~$r$ is {\em essentially unique} if it is unique up to this $GL_{r-n}$  action.

Let us fix 3 distinct indices $k,l,m \leq p$.
Say that a tuple $(A_1,\ldots,A_p)$ of matrices of size $n$ satisfies hypothesis $(H_{klm})$
 if the triple $(A_k,A_l,A_m)$ satisfies the following two conditions: % the hypotheses of Theorem~\ref{th:3unique}, i.e., 
\begin{itemize}
\item[(i)] The three linear spaces $\Ima [A_k,A_l]$, $\Ima [A_k,A_m]$, $\Ima[A_l,A_m]$ are of dimension $2(r-n)$.
\item[(ii)] The three linear spaces $\Ima [A_k,A_l]+\Ima [A_k,A_m]$, $\Ima [A_l,A_k]+\Ima [A_l,A_m]$ and 
$\Ima [A_m,A_k]+\Ima [A_m,A_l]$ are of dimension $3(r-n)$.
\end{itemize}
%Note that the three indices must be distinct if $(H_{klm})$ holds.

%$\dim V_{klm} = 3(r-n)$ (recall from Proposition~\ref{prop:3direct} that $V_{klm}$ is the linear space $\Ima [A_k,A_l]+\Ima [A_k,A_m]$). 
\begin{theorem} \label{th:3unique}
%Consider a tuple  $(A_1,A_2,A_3)$ of matrices of size~$n$ with entries in a field $K$ such that:
%\begin{itemize}
%\item[(i)] The three linear spaces $\Ima [A_1,A_2]$, $\Ima [A_1,A_3]$, $\Ima[A_2,A_3]$ are of dimension $2(r-n)$.
%\item[(ii)] The three linear spaces $\Ima [A_1,A_2]+\Ima [A_1,A_3]$, $\Ima [A_2,A_1]+\Ima [A_2,A_3]$ and 
%$\Ima [A_3,A_1]+\Ima [A_3,A_2]$ are of dimension $3(r-n)$.
%\end{itemize}
Let $(A_1,A_2,A_3)$ be a tuple of matrices of size~$n$ satisfying $(H_{123})$.
This tuple does not have any commuting extension of size less than $r$.
If $(A_1,A_2,A_3)$ has a commuting extension of size $r$, it is essentially unique. Moreover, if a commuting extension of size $r$ exists
in the algebraic closure $\overline{K}$,  there is already one in the ground field $K$.  
\end{theorem}
Uniqueness results for more than 3 matrices can be obtained under various combinations of the $(H_{klm})$.
%%%%%%%%%%{\bf state other versions?}
For instance: 
\begin{theorem} \label{th:uniquext}
Consider a tuple  $(A_1,\ldots,A_p)$ of matrices of size~$n$ with $p \geq 3$ such that there is for all $2 \leq l \leq p$ some $ m \not \in \{1,l\}$ satisfying hypothesis $(H_{1lm})$. %such that  $(H_{1lm})$ holds.
If $(A_1,\ldots,A_p)$ has a commuting extension of size $r$, it is essentially unique. 
Moreover, if a commuting extension of size $r$ exists
in the algebraic closure $\overline{K}$,  there is already one in the ground field $K$.
\end{theorem}

\subsection{The new uniqueness theorem.} \label{sec:genuniquestatement}

In this section we state our main uniqueness theorem for tensor decomposition. This is a generalization of the uniqueness theorem in the introduction (Theorem~\ref{th:unique_simple}), and it is proved in Section~\ref{sec:uniqueproof}.
% Recall that the hypotheses $(H_{klm})$ are defined in Section~\ref{sec:uniquext}.
\begin{theorem} \label{th:unique}
Let $T=\sum_{i=1}^r u_i \otimes v_i \otimes w_i$ be a tensor of format $n \times n \times p$ with $p\geq 4$. 
We assume that $T$ satisfies the following properties:
\begin{enumerate}
\item The $w_i$ are pairwise linearly independent.
\item% Let $(T_1,\ldots,T_p)$ be the slices of $T$. The first slice $T_1$ is invertible.
The span $\langle T_1,\ldots,T_p\rangle$ of the slices of $T$ contains an invertible matrix.
\item %Let  $A_i=A^{-1} T_i$ for $i=2,\ldots,p$. 
There is an invertible matrix $A$ in this span such that the tuple of matrices  $(A^{-1}T_2,\ldots,A^{-1}T_p)$ 
satisfies the hypothesis 
of the uniqueness theorem for commuting extensions.\footnote{Recall that the hypothesis of Theorem~\ref{th:uniquext} is as follows: for all $2 \leq l \leq p-1$ there exists some $m {\not \in} \{1,l\}$ such
that hypothesis $(H_{1lm})$ is satisfied by this tuple. The hypotheses $(H_{klm})$ are defined in Section~\ref{sec:uniquext}.}
\end{enumerate}
Then $\rk(T)=r$, and  the decomposition of $T$ as a sum of~$r$ rank one tensors is essentially unique. 
\end{theorem}

The hypotheses of this theorem are a mix of conditions about the slices of the tensor (hypothesis 2 and 3) and about the
vectors in a decomposition of $T$ (hypothesis 1).
One can show that hypothesis 1 in this theorem can be replaced by:
\begin{itemize}
\item[(1')] The linear subspace of matrices spanned by the slices $T$ is of dimension~$r$.
\end{itemize}
Indeed,  by Corollary~\ref{cor:dimspan} hypothesis 1' implies hypothesis 1 (see also  the proof of Corollary~\ref{cor:jennrich} in Section~\ref{sec:unique}).
Working with hypothesis 1' is therefore less general, but perhaps more aesthetically pleasing since the hypotheses of the resulting theorem are purely about the tensor slices.

\subsection{Jennrich's uniqueness theorem and the simultaneous diagonalization algorithm.}
 %corresponding algorithm}
  \label{sec:unique}

Our presentation of Jennrich's uniqueness theorem and of the corresponding decomposition algorithm is based 
on Moitra's book~\cite{moitra18}. %As explained in~\cite{moitra18}, 
This uniqueness theorem can be traced back (in a slightly less general form)   to Harshman~\cite{harshman70}, where it is attributed to Jennrich. 
Harshman's paper also describes a decomposition algorithm, and it is again attributed to Jennrich. That algorithm is 
{\em not} based on simultaneous diagonalization, or on the uniqueness proof. The name ``Jennrich's algorithm'' is sometimes used (including in the present paper) for the simultaneous diagonalization method in Algorithm~\ref{algo:jennrich} and in~\cite{leurgans93,moitra18}, 
but it does not appear to be historically correct.
Another version of the uniqueness theorem appears in a second paper by Harshman~\cite{harshman72}, and the proof seems closer to a simultaneous diagonalization argument. That second paper does not include any algorithmic result, though.
\begin{theorem}[Jennrich's uniqueness theorem] \label{th:jennrich}
Let $T=\sum_{i=1}^r u_i \otimes v_i \otimes w_i$ be a tensor of format $m \times n \times p$ such that:
\begin{itemize}
\item[(i)] The vectors $u_i$ are linearly independent.
\item[(ii)] The vectors $v_i$ are linearly independent.
\item[(iii)] Every pair of vectors in the set $\{w_i;\ 1\leq i \leq r\}$ is linearly independent.
\end{itemize}
Then $\rk(T)=r$, and the decomposition of $T$ as a sum of $r$ rank one tensors is essentially unique.
\end{theorem}
In~\cite{moitra18}, Jennrich's uniqueness theorem as stated in Theorem~\ref{th:jennrich} and the corresponding decomposition algorithm are attributed to~\cite{leurgans93} (their work was apparently independent from Harshman and Jennrich's).
A comparison between Kruskal's uniqueness theorem~\cite{kruskal77} and the earlier work by Harshman and Jennrich can be found in~\cite{ten09}.  As pointed out in Section~\ref{sec:kruskal}, Theorem~\ref{th:jennrich} is a special case of Kruskal's theorem.

\begin{corollary} \label{cor:jennrich}
%Let $T$ be a tensor of format $n \times n \times p$. 
Let $T$ be an order 3 tensor with slices $T_1,\ldots,T_p$, and $r=\rk(T)$.
Let $A$ be in any matrix belonging to the span $\langle T_1,\ldots,T_k \rangle$ of the slices.
 Then  $r \geq \rk(A)$;  if $r=\rk(A)$,
the decomposition of $T$ as a sum of~$r$  tensors of rank~one is essentially unique whenever at least one of these two conditions
holds:
\begin{enumerate}
\item  There exists a decomposition $T=\sum_{i=1}^r u_i \otimes v_i \otimes w_i$ 
where every pair of vectors in the set $\{w_i;\ 1\leq i \leq r\}$ is linearly independent.
\item The linear subspace of matrices spanned by the slices of $T$ is of dimension~$r$.
\end{enumerate}
\end{corollary}
Note that the second condition can only be satisfied if $p \geq r$. 
We have shown in Corollary~\ref{cor:dimspan} that the second condition implies the first one.
%We will see in the proof that this second condition implies the first one.
\begin{proof}%[Proof of Corollary~\ref{cor:jennrich}]
Consider any decomposition $T=\sum_{i=1}^r u_i \otimes v_i \otimes w_i$.
By~(\ref{eq:slices}) we obtain  from this decomposition 
that $A=U^TDV$ where $D$ is a diagonal matrix of size $r$,
%Suppose for instance that $T_1$ is the invertible slice.  
%Since $D_k$ is of size $r$, we must have 
hence $r \geq \rk(A)$.

Assume now that $r=\rk(A)$. It follows again from $A=U^TDV$ that $U$ and $V$ are of rank $r$,
hence  the first two conditions in Theorem~\ref{th:jennrich} are met.
If we assume that the third condition is satisfied as well, we are done.
If condition 2 of the Corollary is satisfied, we are also done since condition~1 must hold by Corollary~\ref{cor:dimspan}.
%Let us assume alternatively that the slices span a subspace of dimension~$r$. This implies that the the diagonal matrices $D_1,\ldots,D_p$
%also span a subspace of dimension~$r$. This implies in turn  that condition~1 is satisfied. Suppose indeed that
%two of the $w_i$ are colinear, for instance $w_1=\lambda w_2$. We have a corresponding relation of proportionality for the
%diagonal matrices: the first diagonal entry of  each $D_j$ is equal to $\lambda$ times its second diagonal entry. Thus $D_1,\ldots,D_p$ would span a subspace of dimension at most $r-1$, a contradiction.
\end{proof}

%{\bf add a version of Jennrich's algorithm (case of an invertible matrix in span of slices).}
The proof of Theorem~\ref{th:jennrich} is algorithmic. The resulting algorithm is based on matrix diagonalization, and appears below as Algorithm~\ref{algo:jennrich}. We present it only  in the special case $m=n=r$,
which is the only one needed in this paper. 
%for SODA style:
\begin{algorithm} \label{algo:jennrich}
\SetAlgoLined
\nonl \textbf{Input:} a tensor $T \in K^{n \times n \times p}$ satisfying the conditions of Theorem~\ref{th:jennrich} for $r=n$.\\ 
\nonl \textbf{Output:} a decomposition $T=\sum_{i=1}^n u_i \otimes v_i \otimes w_i$.
% (this decomposition is guaranteed to be correct under the hypotheses of Theorem~\ref{th:0decomp}.\\

Compute two random random linear combinations 
$$T^{(a)}=\sum_{k=1}^p a_k T_k,\ T^{(b)}=\sum_{k=1}^p b_k T_k$$
of the slices of $T$.\\
 Compute the eigenvalues $\lambda_1,\ldots,\lambda_n$ and eigenvectors $u_1,\ldots,u_n$ of $T^{(a)}{T^{(b)}}^{-1}$.\\
Compute the eigenvalues $\mu_1,\ldots,\mu_n$ and eigenvectors $v_1,\ldots,v_n$ of $({T^{(a)}}^{-1}T^{(b)})^T$.\\
Reorder the eigenvectors and their eigenvalues  to make sure that the corresponding eigenvalues are reciprocal (i.e., $\lambda_i \mu_i=1$).\\
Solve for $w_i$ in the linear system $T=\sum_{i=1}^n u_i \otimes v_i \otimes w_i$, output this decomposition.
%\caption{Jennrich's algorithm for tensors of rank $n$.}
\caption{decomposition by simultaneous diagonalization (sometimes called "Jennrich's algorithm") for tensors of rank $n$.}
\end{algorithm}

%\caption{decomposition by simultaneous diagonalization (sometimes called "Jennrich's algorithm") for tensors of rank $n$.}
%\end{figure}
An analysis of %Jennrich's 
this algorithm can be found in~\cite{moitra18}, where it is called ``Jennrich's algorithm.'' In particular, it can be shown that with high probability over the choice of the coefficients $a_k,b_k$, the matrices at steps 2 and 3 have $n$ distinct eigenvalues each.\footnote{If $K$ is finite, this field should be large enough compared to $n$ in order to apply the Schwartz-Zippel Lemma~\cite{Schw,zippel}. If $K$ is too small,  we can take the coefficients $a_k,b_k$ in a field extension.}  Moreover, these eigenvalues are reciprocal (refer to step 4). In the presentation in~\cite{moitra18}, the computation of the
matrix inverses ${T^{(a)}}^{-1}$, ${T^{(b)}}^{-1}$ is replaced by the computation of the Moore-Penrose inverses.
 As in Theorem~\ref{th:unique}, this allows the treatment of the case $r \leq n$, and of the decomposition 
of rectangular tensors (of format $m \times n \times p$). A presentation which does not appeal explicitly to the Moore-Penrose inverse can be found in~\cite{bhaskara14}.
An optimized version of %Jennrich's decomposition 
the simultaneous diagonalization algorithm (and a detailed complexity analysis) for the special case of 
symmetric tensors can be found in~\cite{KS23b,KS24}.
\begin{remark} \label{rem:jennrich}
Algorithm~\ref{algo:jennrich} runs in polynomial time in the bit model of computation when the vectors $u_i,v_i,w_i$ in the decomposition of $T$ have rational entries. Indeed, the $u_i$ and $v_i$ are the eigenvectors of the matrices that 
we diagonalize at steps 2 and 3. Since these matrices have rational entries, their  eigenvalues  must %therefore
 be rational as well and they can be computed in polynomial time (the rational roots of a polynomial in $\mathbb{Q}[X]$ can be computed in polynomial time).
Then we can compute the eigenvectors by solving a linear system, and it is clear that the rest of the algorithm also runs in polynomial time in the bit model.
\end{remark}

%{\bf On passe à $T_p$ inversible!}

%Let $T$ be a tensor of format $n \times n \times p$ where $p \geq 4$. %Let $T_1,\ldots,T_p \in M_n(K)$ be the slices
%of $T$. Assume that the span $\langle T_1,\ldots,T_p \rangle$ of the slices of $T$ contains an invertible matrix,
%and let $A$ be any such matrix in the span.
%Let $A_i=A^{-1}T_{i+1}$ for $i=1,\ldots,p-1$.
%Assume also that the tuple $(A_1,\ldots,A_{p-1})$ satisfies the hypotheses of the uniqueness theorem for commuting extensions (Theorem~\ref{th:uniquext})
 %for some appropriate integer
% $r \geq n$, i.e., for all $2 \leq l \leq p-1$ there exists some $m {\not \in} \{1,l\}$ such
%that hypothesis $(H_{1lm})$ is satisfied.
 
 %Then $\rk(T) \geq r$, and if $\rk(T) = r$ its decomposition as a sum of~$r$ rank one tensors is essentially unique whenever at least one of these two conditions
%holds:
%\begin{enumerate}
%\item  There exists a decomposition $T=\sum_{i=1}^r u_i \otimes v_i \otimes w_i$ 
%where every pair of vectors in the set $\{w_i;\ 1\leq i \leq r\}$ is linearly independent.
%\item The linear subspace of matrices spanned by the slices of $T$ is of dimension~$r$.
%\end{enumerate}
%\end{theorem}

\subsection{Kruskal's uniqueness theorem.} \label{sec:kruskal}

This uniqueness theorem is based on the following notion of {\em K-rank} of a set of vectors.
\begin{definition}
The Kruskal rank, or K-rank, of a set of vectors $\{u_1,\ldots,u_r\} \subseteq K^n$ is the largest integer $k$ such that every
subset of $k$ vectors in this set is linearly independent.
\end{definition}

\newpage
\begin{theorem}[Kruskal's uniqueness theorem~\cite{kruskal77}] \label{th:kruskal}
For $T=\sum_{i=1}^r u_i \otimes v_i \otimes w_i$, denote by $k_u,k_v$ and $k_w$ the respective K-ranks of the sets
of vectors $\{u_1,\ldots,u_r\}$,  $\{v_1,\ldots,v_r\}$, $\{w_1,\ldots,w_r\}$.
If $2r+2 \leq k_u+k_v+k_w$ then $\rk(T)=r$ and the decomposition of $T$ as a sum of $r$ tensors of rank 1 is essentially unique.
\end{theorem}
Note that this theorem implies Jennrich's uniqueness theorem (Theorem~\ref{th:jennrich}).
For tensors of format $n \times n \times n$, Kruskal's theorem can prove uniqueness up to $r=\lfloor 3n/2 \rfloor-1$. 
This is better than our uniqueness theorem, which applies only up to $r=\lfloor 4n/3 \rfloor$, for all $n \geq 10$ (and also for $n=8$, but not for $n=9$).
On the other hand, our theorem is quantitatively stronger for tensors with a small (e.g., constant) number of slices.
For instance, we obtain uniqueness up to $r=4n/3$ already for $p=4$ slices. Since $k_w \leq p$, for $p=4$ Theorem~\ref{th:kruskal} only proves uniqueness up to $r=n+1$.
In contrast to Jennrich's theorem, no efficient algorithmic version of Kruskal's theorem is known.\footnote{An algorithmic version of Kruskal's theorem can be found in~\cite{lathauwer14} but the authors do not make any claim on the complexity of their algorithm and indeed,  it does not appear to run in polynomial time.}
Kruskal's proof of Theorem~\ref{th:kruskal} is based on his so-called ``permutation lemma''. One can find in~\cite{rhodes10} an alternative proof which takes Jennrich's theorem as its starting point~\cite[Corollary 2]{rhodes10}. Unfortunately, even though the starting point of that proof is algorithmic, this has not led to an efficient algorithmic version of Kruskal's theorem.

\subsection{Proof of the new uniqueness theorem.} \label{sec:uniqueproof}

This section is devoted to  the proof of Theorem~\ref{th:unique}. From the  hypotheses $(H_{1lm})$ we have $\rk [A^{-1}T_2,A^{-1}T_l]=2(r-n)$ for any $3 \leq l \leq p$.
The lower bound $\rk(T) \geq r$ then follows from Strassen's lower bound (Theorem~\ref{th:strassen}) applied to the tensor of format $n \times n \times 3$ with slices $(A,T_2,T_l)$. 
Indeed, this tensor is of rank at most $\rk(T)$ by Corollary~\ref{cor:add}.
Here we just need to apply $(H_{1lm})$ for a single value of $l$, e.g., $l=3$.
We have therefore shown that $\rk(T)=r$, and we turn our attention to the uniqueness.
%We have already proved that~$\rk(T) \geq r$.
%From the assumptions on $T$ we know that the matrix tuple $(A_1,\ldots,A_{p-1})$ satisfies $(H_{12m})$
 %holds for some $m \in \{3,\ldots,p-1\}$. By 
%In order to prove the second part of the theorem, let us assume that $\rk(T)=r$.
%First, we note that condition~2 in Theorem~\ref{th:unique} implies condition~1, 
%just as in the proof of Corollary~\ref{cor:jennrich}. So we just need to prove Theorem~\ref{th:unique} assuming condition~1 holds.
 This is very easy when $r=n$: the uniqueness  is given by Corollary~\ref{cor:jennrich}.
%{\bf On a besoin d'une version avec $T_k$ remplacé par une matrice dans le span.}
For the rest of the proof, we'll assume that $r>n$ (since the span of the slices contains an invertible matrix, $r<n$ is impossible by  Corollary~\ref{cor:jennrich}).

\subsubsection{Uniqueness assuming $w_{i1} \neq 0$ and $\det(T_1) \neq 0$.} \label{sec:0unique}

 By hypothesis there is  a decomposition 
\begin{equation} \label{eq:T1decomp} T=\sum_{i=1}^r u_i \otimes v_i \otimes w_i
\end{equation}
where every pair of vectors in the set $\{w_i;\ 1\leq i \leq r\}$ is linearly independent.
We will first prove Theorem~\ref{th:unique} under the additional assumption that $w_{i1} \neq 0$ for all $i$, and that
$A=T_1$ (this requires $T_1$ to be invertible).
These assumptions will be relaxed in Section~\ref{sec:genunique}.
We will see in Section~\ref{sec:generic} that the assumptions of our uniqueness theorem are generically satisfied up to
 $r=4n/3$, and this is even true for the stronger assumptions that are in effect in the present section.

Consider any other decomposition
\begin{equation} \label{eq:T2decomp} T=\sum_{i=1}^r u_i^{(2)} \otimes v_i^{(2)} \otimes w_i^{(2)}.
\end{equation}
We must show that these two decompositions of $T$ are identical up to permutation of the rank one blocks (note that there is 
no linear independence assumption on the vectors $ w_i^{(2)}$ in the second decomposition).
Proposition~\ref{prop:constr} can be applied to $T$ with $A=T_1$ since the corresponding matrix in~(\ref{eq:span}) is $D=D_1=\diag(w_{11},\ldots,w_{r1}).$ This matrix has rank $r$ as required since its diagonal entries are nonzero 
from our additional assumption.
%Pick a generic linear combination $T_{\lambda}=\sum_{k=1}^p \lambda _k T_k$ as in Lemma~\ref{lem:unique}.
For each of our two decompositions of $T$, Proposition~\ref{prop:constr} therefore  yields a commuting extension of size $r$ 
for the tuple $(T_{1}^{-1} T_1,\ldots,T_{1}^{-1} T_p)$. % w is the here $A'_i=T_{\lambda}^{-1} T_i$
Namely, by Proposition~\ref{prop:constr}.(ii) we obtain from the first decomposition an extension $(Z_1,\ldots,Z_p)$ 
where the $Z_i$ are the slices of a tensor 
\begin{equation} \label{eq:T'decomp}
T'=\sum_{i=1}^r u'_i \otimes v'_i \otimes w_i.
\end{equation}
 %of format $r \times r \times p$; 
Here,  $v'_i \in K^r$ has its first $n$ coordinates equal to those of $v_i$, and $u'_i \in K^r$ has its first $n$ coordinates equal to 
the the $i$-th row of $UT_{1}^{-T}$. From the second decomposition we obtain an extension $(Z'_1,\ldots,Z'_p)$ 
where the $Z'_i$ are the slices of a tensor 
\begin{equation} \label{eq:T3decomp}
T^{(3)}=\sum_{i=1}^r u_i^{(3)} \otimes v_i^{(3)} \otimes w_i^{(2)}.
\end{equation} %of format $r \times r \times p$; 
The vector  $v_i^{(3)} \in K^r$ has its first $n$ coordinates equal to those of $v_i^{(2)}$, and $u_i^{(3)}\in K^r$ has its first $n$ coordinates equal to the the $i$-th row of $U^{(2)}T_{1}^{-T}$. 
Here $U^{(2)}$ is the matrix with $u_i^{(2)}$ as its $i$-th row. Moreover, the first slices $Z_1,Z'_1$ of $T'$ and $T^{(3)}$ are equal to $I_r$ by Proposition~\ref{prop:constr}.(iv).

By hypothesis the tuple $(T_{1}^{-1} T_2,\ldots,T_{1}^{-1} T_p)$ has an essentially unique extension of size $r$, so there exists $M \in GL_{r-n}(K)$ such that $\rho_M(Z'_i)=Z_i$ for all $i \geq 2$. This remains true even for $i=1$ since 
$\rho_M(Z'_1)=\rho_M(I_r)=I_r=Z_1$.

What is the effect of this transformation on the corresponding tensor decomposition~(\ref{eq:T3decomp})? Recall from~(\ref{eq:action}) that $\rho_M$ multiplies each slice from the right by $$D_M=\begin{pmatrix}
 I_n & 0\\
 0 & M
 \end{pmatrix},$$ and from the left by $D_M^{-1}$. By Proposition~\ref{prop:slices} and Corollary~\ref{cor:mult}, multiplying
 the slices form the left only changes the vectors $u_i^{(3)}$ in the decomposition, and likewise a multiplication from 
 the right only changes the $v_i^{(3)}$. Since $\rho_M$ sends $T^{(3)}$ to $T'$ we obtain a second decomposition of 
 $T'$ in addition to~(\ref{eq:T'decomp}):
 \begin{equation}  \label{eq:T4decomp}
T'=\sum_{i=1}^r u_i^{(4)} \otimes v_i^{(4)} \otimes w_i^{(2)}.
\end{equation}
If we denote by $U^{(3)},U^{(4)},V^{(3)},V^{(4)}$ the matrices with respective rows $u_i^{(3)},u_i^{(4)}, v_i^{(3)}$ and $v_i^{(4)}$
then we have:
\begin{equation} \label{eq:dm}  V^{(4)}=V^{(3)}D_M\ \mathrm{ and }\ U^{(4)}=U^{(3)}D_M^{-T}.
\end{equation}
\begin{lemma}
The decomposition of $T'$ as a sum of $r$ rank one tensors is essentially unique.
\end{lemma}
\begin{proof}
This follows from~(\ref{eq:T'decomp}) and Theorem~\ref{th:jennrich} since:
\begin{itemize}
\item[(i)] The vectors $u'_1,\ldots,u'_r$ are linearly independent by Proposition~\ref{prop:constr}.(iii).
\item[(ii)] The vectors $v'_1,\ldots,v'_r$ are linearly independent for the same reason.
\item[(ii)] the vectors $w_i$ in~(\ref{eq:T'decomp}) are pairwise linearly independent (remember that this is an assumption
of Theorem~\ref{th:unique}).
\end{itemize}
\end{proof}
This lemma implies that the rank-one blocks  in our two decompositions~(\ref{eq:T'decomp}) and~(\ref{eq:T4decomp}) are identical up to permutation. The blocks in~(\ref{eq:T2decomp}), (\ref{eq:T3decomp}) and~(\ref{eq:T4decomp}) can therefore be renumbered so that for all $i \in [r]$,
$$u_i^{(4)} \otimes v_i^{(4)} \otimes w_i^{(2)} = u'_i \otimes v'_i \otimes w_i.$$ The rank-one tensor on both sides of this equality is nonzero (this follows for instance from the fact that $T'$ satisfies the hypotheses of Theorem~\ref{th:jennrich}).
Hence there exist scalars $\alpha_i,\beta_i,\gamma_i$ such that $u_i^{(4)} = \alpha_i u'_i$, $v_i^{(4)} = \beta_i v'_i$, 
$w_i^{(2)} = \gamma_i w_i$ and $\alpha_i \beta_i \gamma_i =1$ for all $i$. 
%Let us now compare the first $n$ coordinates of $u_i^{(4)}$ and u'_i$. 
In this series of equalities, those involving the vectors $u'_i$ and $v'_i$ can be rewritten in matrix form as: $U^{(4)} = AU'$, 
$V^{(4)} = BV'$ where $U'$ and $V'$ are the matrices having respectively  the $u'_i$ and $v'_i$ as rows, $A=\diag(\alpha_1,\ldots,\alpha_r)$ and $B=\diag(\beta_1,\ldots,\beta_r)$. From~(\ref{eq:dm}) 
we have $AU'=U^{(3)}D_M^{-T}$ and $BV'=V^{(3)}D_M$. We now focus on the first of these two equalities, and more precisely on the first $n$ columns of the matrices on both sides of the equality.  Multiplying $U^{(3)}$ by $D_M^{-T}$ does not change its first $n$ columns so, from $AU'=U^{(3)}D_M^{-T}$ we can conclude that $AUT_{1}^{-T} = U^{(2)}T_{1}^{-T}$.
Multiplying both sides by $T_{1}^{T}$ we obtain $AU= U^{(2)}$, or in vector form: $u_i^{(2)} = \alpha_i u_i$ for all $i$.
By a similar but simpler\footnote{$T_{1}$ is replaced by the identity matrix.} argument, the equalities $v_i^{(2)} = \beta_i v_i$ can be derived from $BV'=V^{(3)}D_M$.
We have already seen that $w_i^{(2)} = \gamma_i w_i$ for all $i$. Hence we can at last conclude that the decompositions of $T$ in~(\ref{eq:T1decomp}) and ~(\ref{eq:T2decomp}) are identical up to permutation of the rank one blocks.
This completes the proof of Theorem~\ref{th:unique} under the additional assumption that $T_1$ is invertible 
and $w_{i1} \neq 0$ for all $i$.

\subsubsection{General case of the uniqueness theorem.} \label{sec:genunique}

We now lift the assumptions $w_{i1} \neq 0$ and $\det(T_1) \neq 0$ from Section~\ref{sec:0unique} to complete the proof of Theorem~\ref{th:unique}. For this, we'll assume without loss of that $K$ is an infinite field: if Theorem~\ref{th:unique} is true
for some field $K$ then it is also true for all of its subfields. The assumption that $K$ is infinite will allow us to argue
about generic linear combination of slices.

In this section we therefore denote by~$T$ a tensor satisfying the hypotheses of Theorem~\ref{th:unique}, except in 
Lemma~\ref{lem:combi} below (no extra hypothesis on $T$ is required besides those stated in the lemma).
\begin{lemma} \label{lem:combi}
Let $T$ be a tensor of format $n \times n \times p$ of rank  $r$. Let $T'$ be the tensor obtained from $T$ by replacing
the first slice $T_1$ by a linear combination $T'_1=\sum_{k=1}^p \lambda_k T_k$ where $\lambda_1 \neq 0$.
Then $\rk(T') = r$, and if the decomposition of $T'$ as a sum of rank one tensors is essentially unique then 
the same is true of $T$. Moreover, for a generic choice of $\lambda$ there is %we have 
a decomposition
$T'=\sum_{i=1}^r u'_i \otimes v'_i  \otimes w'_i$ with $w'_{i1} \neq 0$ for all $i$.
\end{lemma}
\begin{proof}
Consider two decompositions: 
\begin{equation} \label{eq:2decomp}
T=\sum_{i=1}^r u_i^{(1)} \otimes v_i^{(1)} \otimes w_i^{(1)} = \sum_{i=1}^r u_i^{(2)} \otimes v_i^{(2)} \otimes w_i^{(2)}.
\end{equation}
By~(\ref{eq:slices}) the slices of $T$ are given by the formula 
$$T_k = {U^{(1)}}^T D^{(1)}_k V^{(1)} = {U^{(2)}}^T D^{(2)}_k V^{(2)}$$
where $U^{(j)}$ and $V^{(j)}$ have the  $u_i^{(j)}$ and $v_i^{(j)}$ as their respective rows, and
 $D_k^{(j)} = \diag(w_{1k}^{(j)},\ldots,w_{rk}^{(j)}).$
Hence the first slice of~$T'$ is $T'_1 = {U^{(1)}}^T D^{(1)}_{\lambda} V^{(1)}$ where $D^{(1)}_{\lambda} = \sum_{k=1}^p \lambda_k D^{(1)}_k$.
This shows that $T'$ admits the decomposition 
\begin{equation} \label{eq:decomp'1}
T'=\sum_{i=1}^r u_i ^{(1)} \otimes v_i^{(1)} \otimes w_i^{(3)}
\end{equation}
 where $w_{ik}^{(3)}=w_{ik}^{(1)}$ for $k \geq 2$, 
and $w_{i1}^{(3)}=\sum_{k=1}^p \lambda_k w_{ik}^{(1)}$,
showing that $\rk(T') \leq r$ (this follows also from Corollary~\ref{cor:add}). %A similar argument 
Writing  $T_1$ as a linear combination of the slices of $T'$ likewise shows that $\rk(T) \leq \rk(T')$, therefore $\rk(T)=\rk(T')=r$.

We can obtain a second decomposition 
\begin{equation} \label{eq:decomp'2}
T'=\sum_{i=1}^r u_i ^{(2)} \otimes v_i^{(2)} \otimes w_i^{(4)}
\end{equation}
from the second decomposition of $T$.  
If the decomposition of $T'$ is essentially unique, we can renumber the rank one blocks in (\ref{eq:2decomp}), (\ref{eq:decomp'1}), (\ref{eq:decomp'2}) so that 
$u_i ^{(1)} \otimes v_i^{(1)} \otimes w_i^{(3)} = u_i ^{(2)} \otimes v_i^{(2)} \otimes w_i^{(4)}$ for all $i$.
These blocks are {\em exactly} of rank one (they are nonzero) since $\rk(T')=r$.
Hence there exist scaling factors $\alpha_i,\beta_i,\gamma_i$ such that $u_i ^{(1)} = \alpha_i u_i ^{(2)}$,
$ v_i^{(1)} = \beta_i  v_i^{(2)}$,  $w_i^{(3)} = \gamma_i w_i^{(4)}$ and $\alpha_i \beta_i \gamma_i =1$ for all $i$.
Moreover,
$$ \lambda_1 w_{i1}^{(1)} = w_{i1}^{(3)} - \sum_{k=2}^p  \lambda_k w_{ik}^{(1)} = 
 w_{i1}^{(3)} - \sum_{k=2}^p  \lambda_k w_{ik}^{(3)}.$$
 Since  $w_i^{(3)} = \gamma_i w_i^{(4)}$, we obtain 
 $$ \lambda_1 w_{i1}^{(1)} = \gamma_i (w_{i1}^{(4)} - \sum_{k=2}^p  \lambda_k w_{ik}^{(4)}) =
  \gamma_i (w_{i1}^{(4)} - \sum_{k=2}^p  \lambda_k w_{ik}^{(2)}) = \gamma_i \lambda_1  w_{i1}^{(2)},$$
  hence $w_{i1}^{(1)} = \gamma_i  w_{i1}^{(2)}$ since $\lambda_1 \neq 0$. The same proportionality relation holds for 
  the other coordinates of  $w_{i}^{(1)}$ and $w_{i}^{(2)}$ since they are equal to the corresponding coordinates of
  $w_i^{(3)}$ and $w_i^{(4)}$. We conclude that  $w_{i}^{(1)} = \gamma_i  w_{i}^{(2)}$, and we have shown that the two decompositions of $T$ are identical up to permutation of their rank one blocks. 
 %$$ \lambda_1 w_

Finally, consider the decomposition of $T'$ in~(\ref{eq:decomp'1}), which we have renamed 
$\sum_{i=1}^r u'_i \otimes v'_i  \otimes w'_i$ in the lemma's statement. We would like the 
$w_{i1}^{(3)}$ to be nonzero for all $i$. For this, $\lambda$ should avoid the union of the $r$ hyperplanes 
$\sum_{k=1}^p w_{ik}^{(1)} \lambda_k=0$. These are proper hyperplanes: if $w_{ik}=0$ for all $k$, then $w_i=0$.
This is impossible, otherwise $T'$ would be of rank less than $r$. We conclude that the  $w_{i1}^{(3)}$ are generically all  nonzero as required.
\end{proof}
As explained at the beginning of Section~\ref{sec:genunique}, we assume from now on that $T$ satisfies the hypotheses of Theorem~\ref{th:unique}. Recall also that $\rk(T)=r$, as shown at the very beginning of Section~\ref{sec:uniqueproof}.
\begin{lemma} \label{lem:unique}
Let us denote by $T_{\lambda}$ the linear combination of slices $\sum_{k=1}^p \lambda _k T_k$. 
For a generic choice of $\lambda \in K^p$, the tuple of matrices $(A'_1,\ldots,A'_{p-1})$ where $A'_i=T_{\lambda}^{-1} T_{i+1}$
satisfies the hypotheses of the uniqueness theorem for commuting extensions (Theorem~\ref{th:uniquext}).
% and has a commuting extension of size $r$ and this extension is essentially unique. %Moreover, $(A'_1,\ldots,A'_p)$ satisfies the hypotheses of Theorem~\ref{th:uniquext}. 

%The commuting extension of size $r$ is therefore essentially unique.
\end{lemma}
%\begin{claim}
%Let us denote by $T_{\lambda}$ the linear combination of slices $\sum_{k=1}^p \lambda _k T_k$. For a generic choice
%of $\lambda \in K^p$, the tuple of matrices $(A'_1,\ldots,A'_{p-1})$ where $A'_i=T_{\lambda}^{-1} T_i$ satisfies the hypotheses of Theorem~\ref{th:uniquext}. 
%\end{claim}
\begin{proof}
% The existence of the commuting extension follows from Proposition~\ref{prop:constr}.(i) and its proof.
%Indeed, from this proof we know that a commuting extension of size $r$ exists if $\lambda$ avoids a hypersurface of degree $n$ (to ensure the invertibility of $T_{\lambda}$) and also avoids a union of $r$ hyperplanes.

%It remains to establish the essential uniqueness of the commuting extension.
From the hypotheses on $T$ we know that  for all $2 \leq l \leq p-1$ there exists some $m {\not \in} \{1,l\}$ such that:
\begin{itemize}
\item[(i)] The three linear spaces $\Ima [A_1,A_l]$, $\Ima [A_1,A_m]$, $\Ima[A_l,A_m]$ are of dimension $2(r-n)$.
\item[(ii)] The three linear spaces $\Ima [A_1,A_l]+\Ima [A_1,A_m]$, $\Ima [A_l,A_1]+\Ima [A_l,A_m]$ and 
$\Ima [A_m,A_1]+\Ima [A_m,A_l]$ are of dimension $3(r-n)$.
\end{itemize}
Here we have $A_i=A^{-1}T_{i+1}$, and $A \in \langle A_1,\ldots,A_p \rangle$.
We'll show that for a generic choice of $\lambda$,  these properties %equalities 
still hold true if we replace $A_1,A_l,A_m$ by 
$A'_1,A'_l,A'_m$. This will show that  $(A'_1,\ldots,A'_{p-1})$ satisfies the  the hypotheses of Theorem~\ref{th:uniquext}.
%The commuting extension of size~$r$ of $(A'_1,\ldots,A'_{p-1})$ is therefore essentially unique.

For given values of $l$ and $m$, we therefore need to compute the dimensions of 6 linear spaces.
First, we note that the corresponding upper bounds on the dimensions of the respective linear spaces hold true
for {\em any} choice of $\lambda$. Consider indeed the tensor $T'$ with slices %$T_1,\ldots,T_{p-1},T_{\lambda}$. 
$T_{\lambda},T_2,\ldots,T_{p}$. 
This tensor is of rank at most $r$ since the tensor with slices $T_{\lambda},T_1,\ldots,T_p$ is of rank
exactly~$r$ by Corollary~\ref{cor:add}. Hence $\rk [A'_1,A'_l] \leq 2(r-n)$ by  Strassen's lower bound (Theorem~\ref{th:strassen}) applied to the tensor of format $n \times n \times 3$ with slices $(T_{\lambda},T_2,T_{l+1})$.
The same argument applies to the linear spaces $\Ima [A'_1,A'_m]$ and $\Ima[A'_l,A'_m]$. 
The dimension of  $\Ima [A'_1,A'_l]+\Ima [A'_1,A'_m]$ can be upper bounded by $3(r-n)$ with a similar argument:
instead of Strassen's lower bound we apply its variation in Theorem~\ref{th:strassenbis} to the tensor 
with slices $(T_{\lambda},T_2,T_{l+1},T_{m+1})$. This argument also applies to the last two remaining spaces.

Next, we show that the 6 converse inequalities hold for a generic choice of $\lambda$. 
This follows from the lower semicontinuity of
matrix rank. Indeed, since %$\rk [A_1,A_l] = \rk [T_p^{-1}T_1,T_p^{-1}T_l] = 2(r-n)$,
 $\rk [A_1,A_l] = \rk [A^{-1}T_2,A^{-1}T_{l+1}] = 2(r-n)$, this matrix contains a certain nonzero minor of size $2(r-n)$.
For any $\lambda$ such that $T_{\lambda}$ is invertible, if the same minor of $[T_{\lambda}^{-1}T_2,T_{\lambda}^{-1}T_{l+1}]$
is nonzero this matrix has rank at least $2(r-n)$. The nonvanishing of this minor is equivalent to the nonvanishing of a certain polynomial $P(\lambda)$, and this polynomial is not identically 0 since it does not vanish at $\lambda=(a_1,\ldots,a_p)$.
Here, $a_1,\ldots,a_p$ are the coefficients such that $A=a_1T_1+\ldots+a_pT_p$.
We have thus shown that the inequality  $\rk [A'_1,A'_l] \geq 2(r-n)$ generically holds true. The same argument applies
to $\rk [A'_1,A'_m]$ and $\rk [A'_l,A'_m]$. 
The dimensions of the 3 remaining spaces can be lower bounded by $3(r-n)$ with a similar argument (recall indeed from the proof of Theorem~\ref{th:strassenbis} that the dimension of each of these spaces can be written as the rank of an appropriate matrix with $n$ rows and $2n$ columns).
\end{proof}
Let us replace the first slice $T_1$ of $T$ by the linear combination $T_{\lambda}=\sum_{k=1}^p \lambda_k T_k$,
and let $T'$ be the resulting tensor. We claim that for a generic~$\lambda$, this tensor satisfies the hypotheses 
of the special case of the uniqueness theorem established in Section~\ref{sec:0unique}.
%Recall indeed that $\rk(T)=r$ as one of the assumptions of Theorem~\ref{th:unique}.
Indeed, since $\rk(T)=r$ the following properties will hold true generically by Lemma~\ref{lem:combi}: (i) $\rk(T')=r$; (ii) the coefficients $w'_{i1}$ 
in the decomposition $T'=\sum_{i=1}^r u'_i \otimes v'_i  \otimes w'_i$ are all nonzero.
Moreover, by Lemma~\ref{lem:unique} the hypotheses of the uniqueness theorem for commuting extensions (Theorem~\ref{th:uniquext}) will also hold true generically for the tuple $({T'_1}^{-1}T'_2,\ldots,{T'_1}^{-1}T'_2)$.
Therefore we can indeed conclude from Section~\ref{sec:0unique} that the decomposition of $T'$ is essentially unique.
Applying Lemma~\ref{lem:combi} a second time shows that the decomposition of $T$ is essentially unique too, and
this completes the proof of Theorem~\ref{th:unique}.~$\Box$

%{\bf to be continued.}

%{\bf To continue: transform $T^{(3)}$ into $T'$ using essential uniqueness, record effect on the vectors in the decomposition.}

\section{Decomposition algorithms} \label{sec:decomp}

%This essentially solves the open problem at the end of section 3.3 of Moitra's book.
In this section we give two algorithms for overcomplete tensor decomposition. The main one (Algorithm~\ref{algo:decomp}) works under the same hypotheses as our general uniqueness result (Theorem~\ref{th:unique}). It calls Algorithm~\ref{algo:0decomp} as a subroutine.
The latter algorithm works under the (somewhat more restrictive) hypotheses of the uniqueness theorem established in Section~\ref{sec:0unique}.

\subsection{Tensors with %4 slices and
an invertible first slice and $w_{i1}\neq 0$.} \label{sec:0decomp}
%\subsubsection{Tensors with 4 slices} \label{sec:4decomp}

We first give an algorithmic version of the uniqueness result in Section~\ref{sec:0unique}.
Namely, we give a decomposition algorithm %for tensors of format $n \times n \times p$ with $p \geq 4$.
%We assume
under the assumption that the first slice $T_1$ is invertible, and that the first entries of the vectors $w_1,\ldots,w_r$ occuring in a rank-$r$ decomposition of the input tensor $T$
are all nonzero. This hypothesis will be relaxed %lifted
 in Section~\ref{sec:gendecomp}. 
 The algorithm goes as follows.

%\begin{figure}
\begin{algorithm} \label{algo:0decomp} %[H]
\SetAlgoLined
\nonl \textbf{Input:} a tensor $T \in K^{n \times n \times p}$ and an integer $r \in [n,4n/3]$.\\ 
\nonl \textbf{Output:} a decomposition $T=\sum_{i=1}^r u_i \otimes v_i \otimes w_i$.\\

Compute the matrices $A_i=T_{1}^{-1} T_i$ for $i=2,\ldots,p$.\\
 Compute a commuting extension $(Z_2,\ldots,Z_p)$ of size $r$ of $(A_2,\ldots,A_p)$ with 
the algorithm from~\cite{koi24}.\\
Let $T'$ be the tensor of format $r \times r \times p$ with slices $(I_r,Z_2,\ldots,Z_p)$.
With the simultaneous diagonalization algorithm (Algorithm~\ref{algo:jennrich}), % (Jennrich's algorithm), 
\newline
compute a decomposition $T'=\sum_{i=1}^r u'_i \otimes v'_i \otimes w_i.$\\
Let $U' \in M_{r,n}(K)$ be the matrix having as $i$-th row the first $n$ coordinates of $u'_i$.
Output the decomposition $T=\sum_{i=1}^r u_i \otimes v_i \otimes w_i$ where $u_i$ is the $i$-th row of $U'T_1^T$ 
and $v_i=(v'_{i1},\ldots,v'_{in})$.
\caption{Decomposition of tensors with  invertible first slice and $w_{i1}\neq 0$.}
\end{algorithm}

%\newpage
\begin{theorem} \label{th:0decomp}
Let $T=\sum_{i=1}^r u_i \otimes v_i \otimes w_i$ be a tensor of format $n \times n \times p$ with $p\geq 4$. 
We assume that $T$ satisfies the following properties:
\begin{enumerate}
\item The $w_i$ are pairwise linearly independent, and $w_{i1} \neq 0$ for all $i$.
\item Let $(T_1,\ldots,T_p)$ be the slices of $T$. The first slice $T_1$ is invertible.
\item Let  $A_i=T_{1}^{-1} T_i$ for $i=2,\ldots,p$. The tuple of matrices  $(A_2,\ldots,A_p)$ satisfies the hypothesis 
of the uniqueness theorem for commuting extensions (Theorem~\ref{th:uniquext}).
\end{enumerate}
Then $rank(T)=r$, and if we run the above algorithm on  input $T$ it will output an (essentially unique) decomposition of rank $r$.
\end{theorem}
\begin{proof}
%We have the upper bound $\rk(T) \leq r$ by hypothesis on $T$. The converse inequality follows from Strassen's lower bound
%applied to the tensor with slices $(T_1,T_2,T_3)$.
%Indeed, as part of the hypotheses of Theorem~\ref{th:uniquext} we have $\rk [A_2,A_3] = 2(r-n)$. 
%This argument was already 
%used at the beginning of the proof of Theorem~\ref{th:unique}. 
%In fact, this theorem implies that the rank $r$ decomposition of $T$ is essentially unique (note that $T_1$ plays the same role in Theorem~\ref{th:0decomp}  as $T_p$ in Theorem~\ref{th:unique}).
Since $T_1$ is invertible, $r \geq n$ by Corollary~\ref{cor:jennrich}.
It follows from Theorem~\ref{th:unique} (or just from the special case treated in Section~\ref{sec:0unique}) 
that $\rk(T)=r$ and that the decomposition of $T$ is essentially unique.
We can apply Proposition~\ref{prop:constr} to $T$ with $A=T_1$: 
by hypothesis, $T_1$ is invertible  and $D_1=\diag(w_{11},\ldots,w_{r1})$ has rank $r$ since its diagonal entries are all nonzero.

By Proposition~\ref{prop:constr}.(i) the tuple $(A_2,\ldots,,A_p)$ has %an essentially 
 a commuting extension of size $r$, and we can find one at step 2 of the algorithm with the algorithm from~\cite{koi24}. Indeed, this algorithm requires the same hypotheses as the uniqueness theorem for commuting extensions (Theorem~\ref{th:uniquext}),  
and $(A_2,\ldots,A_p)$ satisfies these hypotheses. 
The commuting extension is essentially unique by Theorem~\ref{th:uniquext}.

By Proposition~\ref{prop:constr}.(ii)  and (iv), the tensor $T'$ at step 3 has rank $r$  and we can decompose it with 
%Jennrich's 
the simultaneous diagonalization algorithm.
Indeed, $T'$ satisfies the hypotheses of Theorem~\ref{th:jennrich} by Proposition\~ref{prop:constr}.(iii) and the first hypothesis 
of Theorem~\ref{th:0decomp}. Instead of Theorem~\ref{th:jennrich}, we can also appeal to its Corollary~\ref{cor:jennrich} 
since the first slice $I_r$ of $T'$ is invertible. % {\bf cite the algorithm, not the uniqueness theorem.}
Finally, the correctness of step 4 follows from Corollary~\ref{cor:constr}. % Proposition\ref{prop:constr}.(ii).
\end{proof}
Now that we have proved the correctness of Algorithm~\ref{algo:0decomp}, we show that it runs in polynomial time on tensors with a rational decomposition.
\begin{theorem} \label{th:bit0decomp}
Let $T=\sum_{i=1}^r u_i \otimes v_i \otimes w_i$ be a tensor of format $n \times n \times p$ where   $p\geq 4$ and where
the $u_i,v_i,w_i$ have rational entries. Assume moreover that $T$ satisfies the hypotheses of Theorem~\ref{th:0decomp}.
If we run Algorithm~\ref{algo:0decomp} on input $T$, it will output an (essentially unique) decomposition of rank $r$
in polynomial time in the bit model of computation. 
\end{theorem}
\begin{proof}
The first step of Algorithm~\ref{algo:0decomp} clearly runs in polynomial time. At step 2 we compute a commuting extension with the algorithm from~\cite{koi24}. This algorithm runs in in polynomial in the bit model of computation on matrices with rational entries (see~\cite[Theorems 11 and 13]{koi24}, and the 
remark following Theorem 11).
At step 3 we compute a decomposition $T'=\sum_{i=1}^r u'_i \otimes v'_i \otimes w_i$ with Algorithm~\ref{algo:jennrich}.
The vectors $u'_i,v'_i$ have rational entries by Proposition~\ref{prop:constr}.(ii), and it is an assumption of Theorem~\ref{th:bit0decomp} that the  $w_i$ have rational entries. By Remark~\ref{rem:jennrich}, 
the running time of Algorithm~\ref{algo:jennrich} will therefore be polynomial in the bit size of $T'$.
In particular the bit size of the $u'_i$ and $v'_i$ must be polynomially bounded in the bit size of $T'$,
which is polynomially bounded in the bit size of $T$. Finally, step 4 clearly runs in polynomial time.
\end{proof}
Note that in the above theorem the algorithm's complexity is measured (as is standard in the theory of algorithms) 
as a function of the bit size of the input ($T$). In particular, under the assumptions of Theorem~\ref{th:0decomp}
the output size must be polynomially bounded in the output size. 
This is not obvious from the expression $T=\sum_{i=1}^r u_i \otimes v_i \otimes w_i$.
We will present a more general version of this result in Section~\ref{sec:gendecomp}.

\subsection{Tensors with $p \geq 4$ slices: general case.} \label{sec:gendecomp}

Here we lift the hypothesis $w_{i1} \neq 0$ from Section~\ref{sec:0decomp}, and we relax the assumption that $T_1$ is 
invertible.
 The corresponding decomposition algorithm is a simple reduction to 
the case treated in Section~\ref{sec:0decomp}. % $w_{i1} \neq 0$:
Theorem~\ref{th:gendecomp} shows that this algorithm is correct under the same hypotheses as our general uniqueness theorem (Theorem~\ref{th:unique}).
%\begin{figure}
\begin{algorithm} \label{algo:decomp}
\SetAlgoLined
\nonl \textbf{Input:} a tensor $T \in K^{n \times n \times p}$ and an integer $r \in [n,4n/3]$.\\ 
\nonl \textbf{Output:} a decomposition $T=\sum_{i=1}^r u_i \otimes v_i \otimes w_i$.\\
Replace the first slice $T_1$ of the input tensor $T$ by a random linear combination $\sum_{i=1}^p \lambda_i T_i$, call the resulting tensor $T'$.

Run Algorithm~\ref{algo:0decomp} on input $T'$. Let $T'=\sum_{i=1}^r u_i \otimes v_i \otimes w'_i$
be the decomposition produced by this algorithm.

 Return the decomposition $T=\sum_{i=1}^r  u_i \otimes v_i \otimes w_i$ where $w_{ik}=w'_{ik}$ for $k \geq 2$, and $w_{i1}=(w'_{i1}-\sum_{k=2}^p \lambda_k w'_{ik})/\lambda_1$.
\caption{Main algorithm for overcomplete tensor decomposition.}
\end{algorithm}

%\newpage
\begin{theorem} \label{th:gendecomp}
Let $T=\sum_{i=1}^r u_i \otimes v_i \otimes w_i$ be a tensor of format $n \times n \times p$ with $p\geq 4$. 
We assume that $T$ satisfies the following properties:
\begin{enumerate}
\item The $w_i$ are pairwise linearly independent.
\item% Let $(T_1,\ldots,T_p)$ be the slices of $T$. The first slice $T_1$ is invertible.
The span $\langle T_1,\ldots,T_p\rangle$ of the slices of $T$ contains an invertible matrix.
\item %Let  $A_i=A^{-1} T_i$ for $i=2,\ldots,p$. 
There is an invertible matrix $A$ in this span such that the tuple of matrices  $(A^{-1}T_2,\ldots,A^{-1}T_p)$ 
satisfies the hypothesis 
of the uniqueness theorem for commuting extensions (Theorem~\ref{th:uniquext}).
\end{enumerate}
Then $rank(T)=r$, and if we run the above algorithm on  input $T$ it will output an (essentially unique) decomposition of rank $r$.
\end{theorem}
\begin{proof}
%We already know from Theorem~\ref{th:0decomp} that $\rk(T)=r$ since that part of the proof did not use the hypothesis 
%$w_{i1} \neq 0$.
It follows from Theorem~\ref{th:unique} that $\rk(T)=r$ and that the decomposition of $T$ is essentially unique.
 By~(\ref{eq:slices}) the slices of $T$ are given by the formula $T_k = U^T D_k V$.
Hence the first slice of~$T'$ is $T'_1 = U^T D_{\lambda} V$ where $D_{\lambda} = \sum_{k=1}^p \lambda_k D_k$.
This shows that $T'$ admits the decomposition $T'=\sum_{i=1}^r u_i \otimes v_i \otimes w'_i$ where $w'_{ik}=w_{ik}$ for $k \geq 2$, and $w'_{i1}=\sum_{k=1}^p \lambda_k w_{ik}$. Step 3 will therefore return a correct decomposition of $T$
if step 2 returns a correct decomposition of $T'$ and $\lambda_1 \neq 0$.

It remains  to show that $T'$ satisfies the 3 hypotheses of Theorem~\ref{th:0decomp} with high probability. We examine them in order in the remainder of the proof.
\begin{enumerate}
\item We want $w'_{i1}$ to be nonzero for all $i$, so $\lambda$ should the avoid the $r$ hyperplanes 
$\sum_{k=1}^p  w_{ik} \lambda_k = 0$. As in the proof of Proposition~\ref{prop:constr} these are proper hyperplanes: the hypothesis that the $w_i$ are pairwise linearly independent implies in particular that all these vectors are nonzero.
We also need the $w'_i$ to be linearly independent. Consider therefore a pair $(w'_i,w'_j)$ with $i \neq j$, which we
denote $w'_i(\lambda)$ and $w'_j(\lambda)$ to emphasize the dependency on $\lambda$. For $\lambda=(1,0,\cdots,0)$,
$w'_i(\lambda)=w_i$ and $w'_j(\lambda)=w_j$ are linearly independent by hypothesis. This must also be true for a generic
$\lambda$ by lower semicontinuity of matrix rank (the relevant matrix is of size $2 \times p$, with rows  $w'_i(\lambda)$ and $w'_j(\lambda)$). 
\item We also need $T'_1$ to be invertible. Since $T'_1(\lambda)=\sum_{i=1}^p \lambda_i T_i$, this follows in the same
way as above from lower semicontinuity of matrix rank. Namely, $\lambda$ should avoid the zero set of the degree $n$
polynomial $\det(\sum_{i=1}^p \lambda_i T_i)=0$. This polynomial is not identically 0 since $\det A \neq 0$.
\item Finally, the tuple of matrices $(A'_2(\lambda),\ldots,A'_p(\lambda))$ where $A'_i(\lambda)=T'_1(\lambda)^{-1} T_i$ should satisfy the hypothesis of the uniqueness theorem for commuting extensions (Theorem~\ref{th:uniquext}).
By hypothesis on $T$ this is true for $\lambda=(a_1,,\cdots,a_p)$, where $a_1,\ldots,a_p$ are the coefficients in the linear
combination $A=a_1T_1+\ldots+a_pT_p$.
This is therefore true for a generic $\lambda$, as we have already shown in the proof of Lemma~\ref{lem:unique}.
\end{enumerate}
We conclude that  $T'$ indeed satisfies the 3 hypotheses of Theorem~\ref{th:0decomp} for all $\lambda$ outside of the zero set
of a finite number of non identically zero polynomials.
\end{proof}

\begin{remark} \label{rem:random}
In Algorithm~\ref{algo:decomp} we have not specified from what probability distribution the random coefficients
 $\lambda_1,\ldots,\lambda_p$ are drawn. Suppose for instance that they are drawn uniformly at random from a finite set $S$. The proof of Theorem~\ref{th:gendecomp} reveals that these coefficients should avoid the zero sets of polynomially many polynomials of polynomially bounded degree. By the Schwartz-Zippel Lemma~\cite{Schw,zippel}, we can make the probability of error smaller
 than, say, 1/3 (or any other constant) by taking $S$ of polynomial size. We can take $S \subseteq K$ if $K$ is large enough.
 If not, we can take the elements of $S$ from a field extension.
\end{remark}
Now that we have proved the correctness of Algorithm~\ref{algo:decomp}, we show that it runs in polynomial time on tensors with a rational decomposition. For this analysis we will assume that the random coefficients $\lambda_1,\ldots,\lambda_n$ are
drawn from a finite set of integers of the form $S=\{0,1,2,\ldots,P(n,p)\}$ where $P$ is a polynomial.
As pointed out in Remark~\ref{rem:random}, this is enough to make the probability of error smaller than $1/3.$
\begin{theorem} \label{th:bitdecomp}
Let $T=\sum_{i=1}^r u_i \otimes v_i \otimes w_i$ be a tensor of format $n \times n \times p$ where   $p\geq 4$ and where
the $u_i,v_i,w_i$ have rational entries. Assume moreover that $T$ satisfies the hypotheses of Theorem~\ref{th:gendecomp}.
If we run Algorithm~\ref{algo:decomp} on input $T$, it will output an (essentially unique) decomposition of rank $r$
in (randomized) polynomial time in the bit model of computation. 
\end{theorem}
\begin{proof}
Steps 1 and 3 run in polynomial time due to to the assumption  $S=\{0,1,2,\ldots,P(n,p)\}$. Step 2 runs in polynomial time
by Theorem~\ref{th:bit0decomp}.
\end{proof}
We therefore reach a similar conclusion as in Section~\ref{sec:0decomp}.
Namely, if a tensor has a decomposition  $T=\sum_{i=1}^r u_i \otimes v_i \otimes w_i$ where the  $u_i,v_i,w_i$ have rational entries, and $T$ satisfies the hypotheses of Theorem~\ref{th:gendecomp}, then it must have an (essentially unique) decomposition
 $T=\sum_{i=1}^r u_i \otimes v_i \otimes w_i$ where the bit size of the vectors $u_i,v_i,w_i$ is polynomially bounded
 in the bit size of $T$.

\section{Genericity Theorem} \label{sec:generic}

In this section we prove (in Theorem~\ref{th:generic}) that the hypotheses for our uniqueness theorem and decomposition algorithm generically hold true up to $r=4n/3$. This result builds on the genericity theorem for commuting extensions (Theorem~\ref{th:genext}) established in~\cite{koi24}.

Recall that we first gave in Section~\ref{sec:0unique} and Section~\ref{sec:0decomp} versions of this uniqueness theorem
and decomposition algorithm that need somewhat stronger hypotheses than their final versions (which appear in Theorem~\ref{th:unique} and Theorem~\ref{th:gendecomp}).  Theorem~\ref{th:generic} shows that even 
these stronger hypotheses are satisfied generically, and this holds true also for the ``simple form'' of uniqueness presented in the introduction (Theorem~\ref{th:unique_simple}).
\begin{theorem}[Genericity Theorem for Commuting Extensions] \label{th:genext}
 Let $Z_i=R^{-1}D_iR$ where $R \in GL_r(K)$ and where  $D_1,\ldots,D_{p-1}$ are diagonal matrices of size $r$.
 Let $A_i$ be the top left block of size $n$ of $Z_i$.
 The two following properties hold for a generic choice of $R$ and of the $D_i$:
 \begin{itemize}
 \item[(i)] If $p-1 \geq 2$ and $r \leq 3n/2$, $\dim \Ima  [A_k,A_l] = 2(r-n)$ for all $k \neq l$.
 \item[(ii)] If $p-1 \geq 3$ and $r \leq 4n/3$, $\dim (\Ima  [A_k,A_l] + \Ima [A_k,A_m] ) = 3(r-n)$ for any triple of distinct matrices
 $A_k,A_l,A_m$.
 \end{itemize}
 \end{theorem}
 
 \begin{remark} \label{rem:generic}
 With the same notations as in the above theorem, let $U^T$ be the matrix made of the first $n$ rows of $R^{-1}$ and let
 $V$ be made up of the first $n$ columns of $R$. Then $U^TV=I_n$, and the top left block of $Z_k$ is equal to $U^TD_kV$.
 Not coincidentally, this is the  expression given in Proposition~\ref{prop:slices} for the slices of a tensor of rank at most $r$.
 \end{remark}

\begin{theorem}[Genericity Theorem for Tensor Decomposition] \label{th:generic}
The two following properties hold true for generic vectors $u_1,\ldots,u_r \in K^n$, $v_1,\ldots,v_r \in K^n$, $w_1,\ldots,w_r \in K^p$.
\begin{enumerate}
\item If If $r \geq n$, the slices  $T_1,\ldots,T_p \in M_n(K)$ of $T=\sum_{i=1}^r u_i \otimes v_i \otimes w_i$ are all invertible.
\item Assume moreover that $p \geq 4$, $n \leq r \leq 4n/3$ and let $A_k = {T_1}^{-1} T_{k+1}$ for $k=1,\ldots,p-1$. 
%$A_k = T_k T_{p}^{-1}$ for $k=1,\ldots,p-1$. 
Then all the hypothesis $(H_{klm})$ 
as defined in Section~\ref{sec:uniquext}  holds true, %any triple of distinct integers  $k,l,m \leq p-1$.
namely,  $\rk [A_k,A_l]=2(r-n)$ and  $dim(\Ima [A_k,A_l]+\Ima [A_k,A_m])=3(r-n)$ for any triple of distinct integers
$k,l,m \leq p-1$.
% the hypotheses of Theorem~\ref{th:3unique}, i.e.,
%the three linear spaces $\Ima [A_k,A_l]$, $\Ima [A_k,A_m]$, $\Ima[A_l,A_m]$ are of dimension $2(r-n)$,
%and the three linear spaces $\Ima [A_k,A_l]+\Ima [A_k,A_m]$, $\Ima [A_l,A_k]+\Ima [A_l,A_m]$,
%$\Ima [A_m,A_k]+\Ima [A_m,A_l]$ are of dimension $3(r-n)$.
\end{enumerate}
\end{theorem}
\begin{proof}
As usual we start from the formula $T_k=U^TD_kV$ from Proposition~\ref{prop:slices}. The first property is
equivalent to the non-vanishing of all $\det(U^TD_kV)$, viewed as polynomials in the entries of the vectors $u_k,v_k,w_k$.
To show this it suffices to find a single point where none of these polynomials evaluates to~0.
We can take for instance $w_{ki}=1$ for all $i$ and $k$, $u_{ki}=v_{ki}=\delta_{ik}$ for $i \leq n$ and $u_{ki}=v_{ki}=0$ for 
$n<i\leq r$. Then we have $T_k=I_n$ for all $k$.

The upper bound $\rk [A_k,A_l] \leq 2(r-n)$ is Strassen's lower bound (Theorem~\ref{th:strassen}). Note that we do not need 
any genericity assumption for this (the invertibility of $T_1$ suffices). Theorem~\ref{th:strassenbis} yields the
upper bound $\dim (\Ima [A_k,A_l]+\Ima [A_k,A_m]) \leq 3(r-n)$. 
It remains to show that the inequalities
 $$\rk [A_k,A_l] \geq 2(r-n),\ \dim (\Ima [A_k,A_l]+\Ima [A_k,A_m]) \geq 3(r-n)$$ generically holds true.
% since the remaining inequalities will follow in a similar way.
 By lower semi-continuity of matrix rank it suffices to find a single tensor where these inequalities hold true, and
they will be true generically. Note  indeed that the left hand side of the second inequality can be written down as the rank of
a matrix having its first $n$ columns equal to $ [A_k,A_l]$, and its remaining $n$ columns equal to $[A_k,A_m]$.

We will show the existence of such a tensor with the additional property that $U^TV=I_n$ and $w_{k1}=1$ for all $k$.
Note that such a tensor has its first slice equal to $I_n$, so our two inequalities reduce to $\rk [T_k,T_l] \geq 2(r-n)$ and
$\dim (\Ima [T_k,T_l]+\Ima [T_k,T_m]) \geq 3(r-n)$. 
Pick an invertible matrix $R$ and $p-1$ diagonal matrices $D_2,\ldots,D_{p}$, all of size $r$, which satisfy the conclusion 
of Theorem~\ref{th:genext}. Let $D_1=I_r$ and let $U$ and $V$ be defined from $R$ as in Remark~\ref{rem:generic}.
The tensor with slices $T_k=U^T D_k V$ for $1 \leq k \leq p$ satisfies the required properties.
\end{proof}

\section{Worst-case complexity of commuting extensions} \label{sec:worst}

Computing the rank of order 3 tensors is NP-hard. 
Given the close connection between tensor rank and commuting extensions, it is natural to ask whether computing the
size of the smallest commuting extension of a tuple of matrices  is NP-hard. Recall that this problem comes in two flavors: one can consider arbitrary commuting extensions, or only those that are diagonalizable. We denote by {\sc COMM-DIAG-EXTENSION} the second version of this problem, and we show that it is indeed NP-hard using a recent construction by Shitov~\cite{shitov25}.
\begin{theorem} \label{th:nphard}
{\sc COMM-DIAG-EXTENSION} is NP-hard over any field.
\end{theorem}
This theorem follows immediately from Proposition~\ref{prop:reduc} and the fact that computing the tensor rank is NP-hard over any field~\cite{hastad90,schaefer16,shitov16}.
\begin{proposition} \label{prop:reduc}
Given a tensor $T \in K^{m \times n \times p}$ with slices $T_1,\ldots,T_p$, let $A_1,\ldots,A_p$ be the square matrices 
of size $m+n$ defined as:
\begin{equation} \label{eq:padslice}
A_i = \begin{pmatrix}
0_{m \times m} & T_i \\
0_{n \times m} & 0_{n \times n}
\end{pmatrix}.
\end{equation}
We have $s=\rk(T)+m+n$, where
$s$ denotes the size of the smallest commuting extension of $(A_1,\ldots,A_p)$ with diagonalizable matrices.
\end{proposition}
It is crucial here to consider only commuting extensions that are diagonalizable since, as  pointed out in~\cite{shitov25}, the matrices $A_i$ already commute ($A_iA_j=A_jA_i=0$).
The proof of Proposition~\ref{prop:reduc} relies on the next two lemmas.
\begin{lemma}[Shitov \cite{shitov25}, Example 3.3] \label{lem:shitov}
Given a tensor $T \in K^{m \times n \times p}$ with slices $T_1,\ldots,T_p$, define $T'$ as the tensor of format $(m+n)\times (m+n) \times (p+1)$ which has the identity matrix as first slice, 
and  the matrices $A_1,\ldots,A_p$ in~(\ref{eq:padslice}) as its remaining $p$ slices.
Then $\rk(T')=\rk(T)+m+n$.
%with its $p$ remaining slices of the form 
%$$\begin{pmatrix}
%0_{m \times m} & T_i \\
%0_{n \times m} & 0_{n \times n}
%\end{pmatrix}.$$ Then $\rk(T') =  \rk(T)+m+n.$
\end{lemma}

\begin{lemma} \label{lem:1moreslice}
Let $T$ and $T'$ be as in the previous lemma, and let $r = \rk(T)$. The tensor $T'$ has a decomposition 
$$T'=\sum_{i=1}^{r+m+n} u'_i \otimes v'_i \otimes w'_i$$ with $w'_{i1} =1$ for all $i$.
\end{lemma}
\begin{proof}
Let $T''$ be the tensor of format $(m+n) \times (m+n) \times p$ with slices $A_1,\ldots,A_p$.
From any decomposition $T=\sum_{i=1}^r u_i \otimes v_i \otimes w_i$, we can obtain a decomposition 
$T''=\sum_{i=1}^r u'_i \otimes v'_i \otimes w_i$ where the vectors $u'_i,v'_i$ are obtained from $u_i,v_i$ by padding with 0's.
Let $w'_i \in K^{p+1}$ be the vector  obtained from $w_i$ by adding a 1 in first position.
The tensor $\sum_{i=1}^r u'_i \otimes v'_i \otimes w'_i$ has almost the same slices as $T'$. More precisely, the last $p$ slices are the same; but its first slice ($\sum_{i=1}^r u'_i \otimes v'_i$) might be different from the first slice of $T'$, which is equal to $I_{m+n}$. To fix this problem, we add a tensor of the form $\sum_{i=r+1}^{r+m+n} u'_i \otimes v'_i \otimes e_1$
where the $m+n$ matrices $u'_i \otimes v'_i$ are chosen so that $$\sum_{i=r+1}^{r+m+n} u'_i \otimes v'_i=I_{m+n}-\sum_{i=1}^r u'_i \otimes v'_i.$$
This corrects the first slice, and the other slices are unchanged. We therefore obtain the desired decomposition for $T'$.
\end{proof}
In view of Lemma~\ref{lem:shitov}, Proposition~\ref{prop:reduc} will be established if we can show that $s=\rk(T')$.
This follows from the characterization of tensor rank by commuting extensions in Theorem~\ref{th:constr} and Proposition~\ref{prop:constr}.
Indeed, one direction of Theorem~\ref{th:constr} shows that $\rk(T') \leq s$ (apply the theorem with $A$ equal to the first slice  of $T'$, i.e., $A=I_{m+n}$).
To obtain the converse inequality $s \leq \rk(T')$, we apply Proposition~\ref{prop:constr} to the decomposition of Lemma~\ref{lem:1moreslice}. Indeed, we can take again $A=I_{m+n}$ 
since the corresponding matrix  in~(\ref{eq:span}) has full rank: $D=\diag(w_{11},\ldots,w_{r+m+n,1})=I_{r+m+n}$.
This completes the proofs of Proposition~\ref{prop:reduc} and Theorem~\ref{th:nphard}.

Finally, we point out that the complexity of tensor rank for a constant number of slices (and the complexity of  {\sc COMM-DIAG-EXTENSION} for a constant number of matrices) seems to be unknown. For instance, is it NP-hard to compute the rank of an order~3 tensor %with 3 or 4 slices?
of format $n \times n \times 3$, or $n \times n \times 4$?

\small

\section*{Acknowledgments} I would like to thank Yaroslav Shitov for discussions on Theorem~\ref{th:nphard}.
 The anonymous referees for SODA 2025 made several useful suggestions, in particular regarding 
the bibliographic references.
%The anonymous referees suggested several improvements in the presentation and some important references from the vast literature on tensor rank.

%\bibliographystyle{plain}
%\bibliography{../../Biblio/biblio}

\end{document}